\documentclass{IEEEtran}

\usepackage{comment}
\usepackage{algorithm,algorithmicx,algpseudocode,amsmath,amssymb,amsthm}
\usepackage{bm,color,dsfont,pifont,subfigure,hyperref}
\usepackage{makecell}
\usepackage{multirow}
\usepackage{scalerel}
\usepackage{url}

\usepackage{authblk}

\makeatletter
\renewcommand{\maketag@@@}[1]{\hbox{\m@th\normalsize\normalfont#1}}%
\makeatother

\newtheorem{definition}{Definition}

\newtheorem{lemma}{Lemma}
\newtheorem{theorem}{Theorem} 
\newtheorem{remark}{Remark}

\newtheorem{claim}{Claim}
\newtheorem{example}{Example}

\allowdisplaybreaks
\newcommand\numberthis{\addtocounter{equation}{1}\tag{\theequation}}

\newcommand{\norm}[1]{\left\lVert#1\right\rVert}
\renewcommand{\(}{\left(}
\renewcommand{\)}{\right)}

\def\la{\left\langle}
\def\ra{\right\rangle}
\def\lcb{\left\{}
\def\rcb{\right\}}
\def\ln{\left\|}
\def\rn{\right\|}

\def\A{\mathcal{A}}
\def\C{\mathbb{C}}

\def\P{\mathcal{P}}
\def\T{\mathcal{T}}
\def\H{\mathcal{H}}
\def\D{\mathcal{D}}
\def\G{\mathcal{G}}
\def\R{\mathbb{R}}
\def\E{\mathbb{E}}

\def\CS{\mathcal{C}}
\def\I{\mathcal{I}}

\newcommand{\ve}{\bm{e}}

\newcommand{\vx}{\bm{x}}

\newcommand{\vw}{\bm{w}}
\newcommand{\vy}{\bm{y}}
\newcommand{\vz}{\bm{z}}

\newcommand{\vX}{\bm{X}}

\newcommand{\vZ}{\bm{Z}}
\newcommand{\vA}{\bm{A}}
\newcommand{\vB}{\bm{B}}
\newcommand{\vC}{\bm{C}}
\newcommand{\vD}{\bm{D}}
\newcommand{\vE}{\bm{E}}
\newcommand{\vF}{\bm{F}}
\newcommand{\vH}{\bm{H}}
\newcommand{\vL}{\bm{L}}
\newcommand{\vR}{\bm{R}}
\newcommand{\vW}{\bm{W}}
\newcommand{\vS}{\bm{S}}
\newcommand{\vM}{\bm{M}}
\newcommand{\vI}{\bm{I}}
\newcommand{\vU}{\bm{U}}
\newcommand{\vV}{\bm{V}}

\newcommand{\vP}{\bm{P}}
\newcommand{\vQ}{\bm{Q}}

\def\vSS{\bm{\Sigma}}
\newcommand \vDelta{\bm{\Delta}}
\newcommand \vDeltaa{\bm{\Delta_1}}
\newcommand \vDeltab{\bm{\Delta_2}}

\DeclareMathOperator{\rank}{rank}
\DeclareMathOperator{\diag}{diag}
\DeclareMathOperator{\Real}{{\normalfont Re}}

\newcommand{\PC}[1]{\P_{\CS}(#1)}

\newcommand{\dist}[2]{{\normalfont\mbox{dist}}(#1,#2)}
\newcommand{\distQ}[2]{{\normalfont\mbox{dist}_Q}(#1,#2)}
\newcommand{\distP}[2]{{\normalfont\mbox{dist}_P}(#1,#2)}
\newcommand{\distPsq}[2]{{\normalfont\mbox{dist}^2_P}(#1,#2)}

\newcommand{\kw}[1]{#1}

\newcommand{\js}[1]{{\color{blue} #1}}
\makeatletter
\renewcommand*{\@fnsymbol}[1]{\ensuremath{\ifcase#1\or \dagger\or \ddagger\or \mathsection\or
    *\or \mathparagraph\or \|\or **\or \dagger\dagger
    \or \ddagger\ddagger \else\@ctrerr\fi}}
\makeatother

\title{Projected Gradient Descent for Spectral Compressed Sensing via Symmetric Hankel Factorization}
\author{Jinsheng Li, Wei Cui, and Xu Zhang

\thanks{
J. Li and W. Cui are with the School of Information and Electronics, Beijing Institute of
Technology, Beijing 100081, China (e-mail: jinshengli@bit.edu.cn; cuiwei@bit.edu.cn). }
\thanks{X.~Zhang is with the School of Artificial Intelligence, Xidian University, Xi'an 710126, China (e-mail: zhang.xu@xidian.edu.cn). }
\thanks{Corresponding author: Xu Zhang.}

}

\begin{document}


\maketitle
\begin{abstract}
Current spectral compressed sensing methods via Hankel matrix completion employ symmetric factorization to demonstrate the low-rank property of the Hankel matrix. However, previous non-convex gradient methods only utilize asymmetric factorization to achieve spectral compressed sensing. In this paper, we propose a novel nonconvex projected gradient descent method for spectral compressed sensing via symmetric factorization named Symmetric Hankel Projected Gradient Descent (SHGD), which updates only one matrix and avoids a balancing regularization term. SHGD reduces about half of the computation and storage costs compared to the prior gradient method based on asymmetric factorization. {Besides, the symmetric factorization employed in our work is completely novel to the prior low-rank factorization model, introducing a new factorization ambiguity under complex orthogonal transformation}. Novel distance metrics are designed for our factorization method and a linear convergence guarantee to the desired signal is established with $O(r^2\log(n))$ observations. 
Numerical simulations demonstrate the superior performance of the proposed SHGD method in phase transitions and computation efficiency compared to state-of-the-art methods.
\end{abstract}
\section{Introduction}
In this paper, we study the spectrally compressed sensing problem, which aims to recover a spectrally sparse signal from partial measurements. Spectrally compressed sensing widely arises in applications such as target localization in radar systems \cite{Potter2010}, analog-to-digital conversion \cite{Tropp2010}, magnetic resonance imaging \cite{Lustig2007}, and channel estimation \cite{Paredes2007}. Denote the spectrally sparse signal as  $\vx=[x_0,x_1,\cdots,x_{n-1}]^T\in\C^{n}$, whose elements are a superposition of $r$ complex sinusoids
\begin{align*}
x(t) = \sum_{k=1}^rd_ke^{(i2\pi f_k-\tau_k)t},\numberthis\label{eq:signal_model}
\end{align*}
where $t=0,\cdots,n-1$, $f_k\in[0,1)$ is the $k$-th normalized frequency, $d_k$ is the amplitude of the $k$-th sinusoids, and $\tau_k$ is the damping factor.

Only a portion of the spectrally sparse signals can be observed in practice as obtaining the complete set of sampling points is time-consuming and technically challenging due to hardware limitations. Additionally, the canonical full sampling on a uniform grid is inefficient due to the weakening effects caused by the damping factors $\{\tau_k\}$. One natural question is how to recover the entire elements of the desired signal $\vx$ from partial and nonuniform known entries, i.e.,
\begin{align*}
\mbox{Find}\quad\vx\quad\mbox{subject to}\quad\P_{\Omega}(\vx)=\sum_{a\in{\Omega}}x_a\ve_a,\numberthis\label{eq:setup}
\end{align*}
where $\ve_a$ is the canonical basis of $\R^n$, ${\Omega}\subseteq\{0,\cdots,n-1\}$ is the sampling index set with cardinality as $m$, 
 $\P_{\Omega}$ is an operator which projects the signal $\vx$ onto the sampling index set ${\Omega}$, and $\P_{\Omega}(\vz)=\sum_{a\in{\Omega}}\la\vz,\ve_a\ra\ve_a$ for $\vz\in\C^n$.

 To reconstruct the spectrally sparse signals from partial observations, many spectral compressed sensing approaches \cite{Chen2014,Cai2018,Cai2019,Zhang2021} were proposed to
 exploit the spectral sparsity via the low-rankness of the lifted Hankel matrix $\H\vx$, where $\H:\C^n\rightarrow\C^{n_1\times n_2}$ is the Hankel lifting operator, and thus $\H\vx$ is an $n_1\times n_2~(n=n_1+n_2-1)$ Hankel matrix. Especially, the low-rankness of the lifted Hankel matrix $\H\vx$ arises from the following Vandermonde decomposition
$$
\mathcal{H}\bm{x}=\bm{E}_L\bm{D}\bm{E}_R^T,
$$
where $\vE_L$ is an $n_1\times r$ Vandermonde matrix whose $k$-th column is $[1,w_k,\cdots,w_k^{n_1-1}]^T$, $\vE_R$ is an $n_2\times r$ Vandermonde matrix whose $k$-th column is $[1,w_k,\cdots,w_k^{n_2-1}]^T$, $w_k=e^{(i2\pi f_k-\tau_k)}$ and $\bm{D}=\diag(d_1,\cdots,d_r)$. When the frequencies $\{f_k\}$ are distinct and each diagonal entry $d_k$ is nonzero, we have $\rank(\H\vx)=r$.
 

Using the low rankness of the lifted Hankel matrix, a convex relaxation approach named Enhanced Matrix Completion (EMaC) was proposed in \cite{Chen2014} to reconstruct spectrally sparse signals. However, the convex approach faces a computational challenge when dealing with large-scale problems. To address the computational challenge, several non-convex approaches were proposed, including PGD \cite{Cai2018}, FIHT \cite{Cai2019}, and ADMM with prior information \cite{Zhang2021}, which demonstrate a lower computation cost than convex approaches \cite{Chen2014,Tang2013}, especially in high-dimensional situations. In particular, PGD is a non-convex projected gradient descent method that applies asymmetric low-rank factorization to parameterize the lifted low-rank Hankel matrix, i.e., $\H\vx=\vZ_U\vZ_V^H$ and introduces a regularization term to reduce the solution space, where $\vZ_U\in\C^{n_1\times r}$ and $\vZ_V\in\C^{n_2\times r}$.


 \begin{figure*}[!t]
		\centering
		\includegraphics[width=0.7\linewidth]{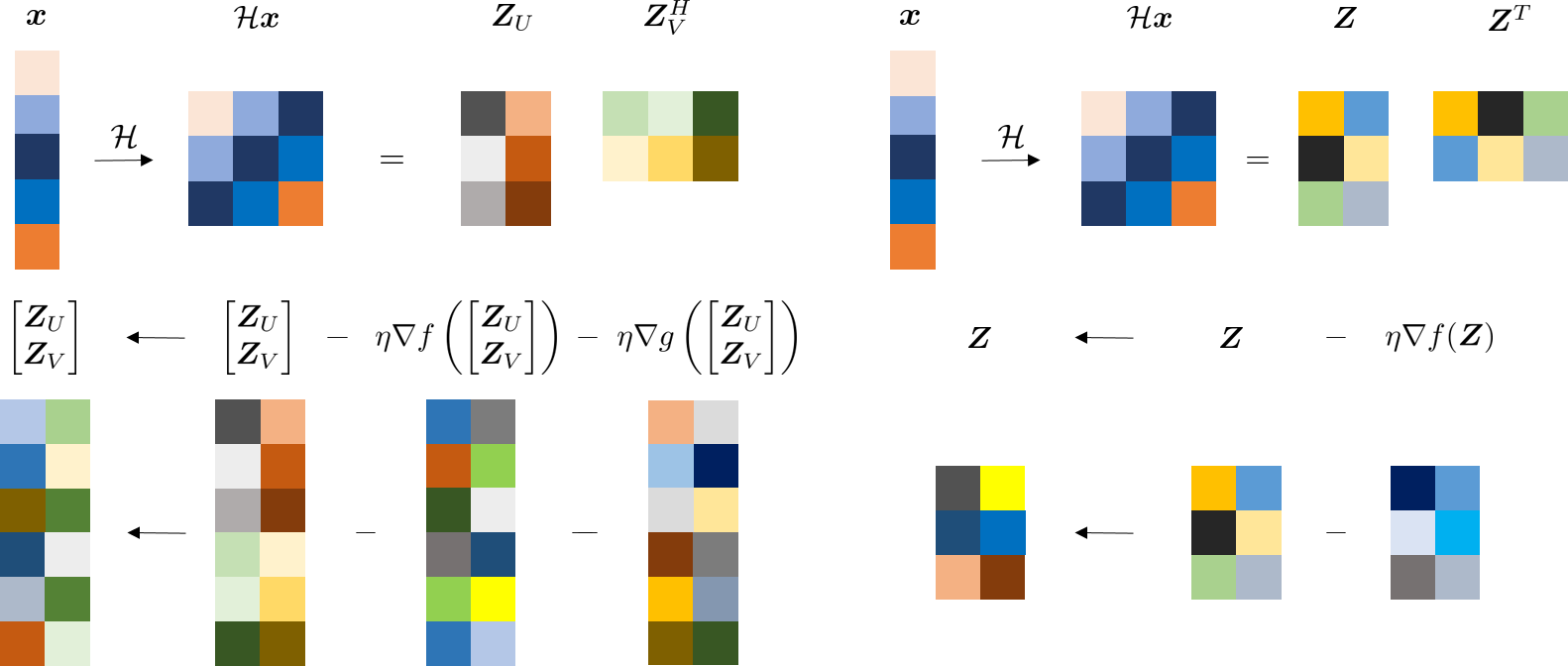}  
\caption{{\textbf{Left}: Asymmetric factorization and update rule of PGD. 
\textbf{Right}: Symmetric factorization and update rule of the proposed SHGD. 
The loss function, $f(\cdot)$, is defined separately in the PGD and SHGD approaches, while the regularization function $g(\cdot)$ is specific to PGD.
SHGD does not need a balancing regularization term and updates one single factor at a time compared to PGD, reducing at least half of operations and storage costs.}}
\label{fig:diagram}
\end{figure*}

\subsection{Motivation and Contributions}
Current spectral compressed sensing methods via Hankel matrix completion employ symmetric factorization to demonstrate the low-rankness of the lifted Hankel matrix\footnote{Without loss of generality, we suppose that the dimension of the signal is odd; otherwise, we can make it odd by either zero padding after the end of the signal or deleting the last entry. 
}. However, these methods employ a non-convex gradient descent algorithm using asymmetric factorization to solve the spectral compressed sensing problem. This raises the question of whether a non-convex gradient method based on symmetric factorization can be used instead.

Fortunately, when the dimension of the signal is odd, the lifted matrix can be a square Hankel matrix $\H\vx\in \C^{n_s\times n_s}~(n=2n_s-1)$ that enjoys the following {complex} symmetric factorization via Vandermonde decomposition
\begin{align*}
\H\vx=\bm{E}_L\bm{D}\bm{E}_R^T=\bm{E}\bm{D}\bm{E}^T=\vZ_\natural\vZ_\natural^T,
\end{align*}
where we can set the $n_s \times r$ Vandermonde matrix $\bm{E}_L=\bm{E}_R \triangleq \vE$  in square case and the $n_s\times r$ matrix $\vZ_\natural \triangleq \vE\vD^{\frac{1}{2}}$. This inspires us to use a low-rank {complex} symmetric factorization $\vM=\H\vz=\vZ\vZ^T$ with the Hankel constraint to estimate $\H\vx$ where $\vZ\in\C^{n_s\times r}$, and further recover the desired signal $\vx$. 
Using symmetric factorization offers two advantages. First, we can avoid the need for a regularization term to balance asymmetric matrices. Second, we can reduce computation time and storage costs by optimizing a single matrix instead of two.

 In this paper, we propose a projected gradient descent method under symmetric Hankel factorization to solve the spectrally compressed sensing problem, named Symmetric Hankel Projected Gradient Descent (SHGD). The differences between SHGD and PGD are shown in Fig. \ref{fig:diagram}, where SHGD reduces at least half of operations and storage costs compared to PGD. Our main contributions are {twofold}:
 
 {
\noindent \textbf{1) }A new nonconvex symmetrically factorized gradient descent method named SHGD is proposed. 
It is demonstrated that SHGD enjoys a linear convergence rate towards the desired spectrally sparse signal with high probability, provided the number of observations is $O(r^2\log(n))$. 
Extensive numerical simulations are conducted to validate the performance of the proposed SHGD. Compared to PGD, SHGD reduces about half of the computational cost and storage requirement. Additionally, SHGD demonstrates superior phase transition behavior compared to FIHT, while maintaining comparable computational efficiency.

\noindent \textbf{2)} 
The complex symmetric factorization $\vM=\vZ\vZ^T$ employed in our work is novel to the prior low-rank factorization model in the literature \footnote{{There are two types of factorization for low-rank real matrices: $\vM=\vL\vR^T$ for asymmetric matrices and $\vM=\vZ\vZ^T$ for positive semi-definite matrices. However, there are three types of factorization for complex matrices:  $\vM=\vL\vR^H$ for asymmetric matrices, $\vM=\vZ\vZ^H$ for Hermitian positive semi-definite matrices, and $\vM=\vZ\vZ^T$ for complex symmetric matrices, which is first pointed out in our work.}}. Complex symmetric factorization introduces a new factorization ambiguity $\vM=\vZ\vZ^T=\vZ\vQ(\vZ\vQ)^T$ where $\vQ$ is called \textit{complex orthogonal matrix} (not unitary) \cite{Horn2012,Horn1999,Luk1997,Didenko1978} such that $\vQ\vQ^T=\vI$ for $\vQ\in \C^{r \times r}$. Such 
factorization ambiguity under complex orthogonal transformation is {first pointed out} in our work. Consequently, the analyses are non-trivial and highly technical to deal with a new factorization model and algorithm. Novel distance metrics are designed to analyze such factorization model.} 

\subsection{Related work}
 When a spectral sparse signal is undamped and its frequencies are discretized on a uniform grid, one can use conventional compressed sensing \cite{Candes2006,Donoho2006} to estimate its spectrum. However, the true frequencies are usually continuous and off-the-grid in practice, thus mismatch errors arise when applying conventional compressed sensing approaches. A grid-free approach named atomic norm minimization (ANM) was proposed by exploiting the sparsity of the frequency domain
 in a continuous way \cite{Tang2013}. When frequencies are well-separated, ANM can achieve exact recovery with high probability from $O(r\log(r)\log(n))$ random samples.  Another grid-free approach was proposed to exploit the sparsity of frequencies via the low-rank Hankel completion model \cite{Chen2014}. By formulating it as a nuclear norm minimization problem, the authors showed $O(r\log^4(n))$ random observations are sufficient to recover the true signals exactly with high probability, equipped with the incoherence conditions. Both of them are convex approaches and suffer from computation issues for large-scale problems.
 
Several provable non-convex approaches were proposed to solve the previous low-rank Hankel matrix completion problem towards spectrally compressed sensing, such as PGD \cite{Cai2018}, FIHT \cite{Cai2019} and pMAP \cite{Shen2022}. These methods offer lower computational costs compared to convex approaches, particularly in high-dimensional scenarios.  Inspired by a factorized projected gradient method towards matrix completion \cite{Zheng2016}, a similar gradient method PGD \cite{Cai2018} that employs low-rank factorization was proposed for Hankel matrix completion. Equipped with $O(r^2\log(n))$ random samples, PGD enjoys an exact recovery with high probability \cite{Cai2018}.  However, the inherent symmetric structure of Hankel matrices wasn't exploited thoroughly in \cite{Cai2018}. The complex symmetry of the square Hankel matrix was exploited in \cite{Xu2008} for fast symmetric SVD of Hankel matrices and in \cite{Andersson2011} for a fast alternating projection method towards complex frequency estimation, but there are no theoretical guarantees in \cite{Xu2008} and \cite{Andersson2011}. FIHT \cite{Cai2019} is an algorithm that employs the Riemannian optimization technique towards low-rank matrices \cite{Wei2016,Wei2020} to solve the Hankel matrix completion problem, and the complex symmetry could be exploited for further reduction of storage and computation. FIHT achieves an exact recovery with high probability equipped with $O(r^2\log^2(n))$ random samples, which is worse with an order of $\log
(n)$ compared to PGD and SHGD. A penalized version of the Method of Alternating Projections (MAP) was proposed in \cite{Shen2022} towards weighted low-rank Hankel optimization, of which the low-rank Hankel matrix completion was a special case. pMAP enjoys a linear convergence rate to the truth with a reasonable computation cost. 
Considering the prior information in the spectral compressed sensing problem, the authors in \cite{Zhang2021} employed low-rank factorization, and proposed an ADMM framework with prior information, but there is no convergence rate analysis for their nonconvex approach. A concurrent work called HT-GD \cite{Wu2023} was proposed to employ the symmetric property of low-rank Hankel matrices and Hermitian property of low-rank Toeplitz matrices, but it focuses on undamped signals and lacks theoretical guarantees \cite{Wu2023}. {Inspired by the previous PGD method \cite{Cai2018}, the authors in \cite{Mao2022} proposed a projected gradient descent method via vectorized Hankel lift (PGD-VHL) towards blind super-resolution. They also apply the approach which is based on low-rank asymmetric factorization with a balancing regularization term.} 

{Symmetric factorization arises in a variety of fields, such as non-negative matrix factorization for clustering \cite{Yin2023,He2011}, matrix completion \cite{Chen2015}, mobile communications, and factor analysis where symmetric tensor factorization is employed \cite{Comon2008,Brachat2010,Cai2015a}, further reducing computation and storage costs. However, the complex symmetric factorization in this work is a new type of factorization, which hasn't appeared in the literature to our knowledge, enriching the range of low-rank matrix factorization approaches. }

\noindent \textbf{Notations}. 
We denote vectors with bold lowercase letters, matrices with bold uppercase letters, and operators with calligraphic letters. For matrix $\vZ$, we use $\vZ^T$, $\vZ^H$, $\bar{\vZ}$, $\norm{\vZ}$, and $\norm{\vZ}_F$ to denote its transpose, conjugate transpose, complex conjugate, spectral norm, and Frobenius norm, respectively. Besides, we define $\norm{\vZ}_{2,\infty}$ as the largest $\ell_2$-norm of its rows. For vector $\vz$, $\norm{\vz}_1$ and $\norm{\vz}_2$ denotes its $\ell_1$ and $\ell_2$ norms. Define the inner product of two matrices $\vZ_1$ and $\vZ_2$ as $\la\vZ_1,\vZ_2\ra=\mathrm{trace}(\vZ_1^H\vZ_2)$.
$[n]$ denotes the set $\{0,1,\cdots,n-1\}$ where $n$ is a natural number. We denote the identity matrix and operator as $\vI$ and $\I$, respectively. The adjoint of the operator $\A$ is denoted as $\A^*$. $\Real(\cdot)$ denotes the real part of a complex number.
  
\section{Problem formulation} \label{sec:pb_formul}
Without loss of generality, we suppose the length $n$ of the desired signal $\vx$ is odd in the following. We aim to recover the desired signal via symmetric Hankel matrix completion. Firstly, we form a $n_s\times n_s$ complex symmetric Hankel matrix 
\begin{align*}
\H\vx =\begin{bmatrix}
x_0 & x_1 & \cdots& x_{n_s-1}\\
x_1 & x_2& \cdots &  x_{n_s}\\
\vdots & \vdots  & \vdots &\vdots\\
x_{n_s-1} & x_{n_s}  &  \cdots  &x_{n-1}
\end{bmatrix},
\end{align*}
where $\H:\C^n\rightarrow\C^{n_s\times n_s}~(n=2n_s-1)$ is the Hankel lifting linear operator.
The lifted symmetric Hankel matrix $\H\vx$ has a Vandermonde decomposition as pointed out in the introduction:
$$
\mathcal{H}\bm{x}=\bm{E}\bm{D}\bm{E}^T,
$$
where $\vE$ is an $n_s\times r$ Vandermonde matrix with its $k$-th column as $[1,w_k,\cdots,w_k^{n_s-1}]^T$ and $w_k=e^{(i2\pi f_k-\tau_k)}$. When frequencies are distinct and each diagonal entry $d_k$ of $\vD$ is non-zero, $\rank(\H\vx)=r$.


 We denote $\H^*$ as the adjoint of  $\H$, which is to map a matrix $\vM\in\C^{n_s\times n_s}$ to a vector $\H^*\vM=\{\sum_{i+j=a}\vM_{(i,j)}\}_{a=0}^{n-1}$. Let $\D^2=\H^*\H$ and $\D$ is a linear operator that maps a vector $\vx\in\C^n$ to $\D\vx\in\C^n$, where $[\D\vx]_a=\sqrt{w_a}x_a$ and $w_a$ is the length of the $a$-th skew-diagonal of an $n_s\times n_s$ matrix with $a\in[n-1]$. Defining $\G=\H\D^{-1}$ as $\D$ is invertible, then we have $\G^*\G=\I$. 
 
As stated in \cite{Cai2018}, a rank constraint {\em weighted least square} problem is constructed to recover $\vx$, utilizing the exact low-rankness  property of $\H\vx$:
\begin{align*}
\min_{\vz\in\C^n}\la \P_{\Omega}\(\D(\vz-\vx)\),\D(\vz-\vx)\ra~~\mbox{s.t.}~\rank(\H\vz)=r.\numberthis\label{eq:low_rank_H}
\end{align*}

The problem \eqref{eq:low_rank_H} can be reformulated as the following optimization problem after making the substitutions that $\vy\leftarrow \D\vx$ and $\vz\leftarrow \D\vz$
\begin{align*}
\min_{\vz\in\C^{n}}\la \P_{\Omega}\(\vz-\vy\),\vz-\vy \ra~\mbox{s.t.}~\rank(\G\vz)=r.\numberthis\label{eq:low_rank_G}
\end{align*}
Note that  $\G\vz\in \C^{n_s\times n_s}$ is a rank-$r$ square Hankel matrix and consequently complex symmetric.

For any complex symmetric matrix $\vM=\vM^T\in\C^{n_s\times n_s}$, there is a special form of SVD called Takagi factorization \cite{Horn2012}, which is defined as $\vM=\vU\vSS\vU^T,$
where $\vU$ is a unitary matrix and $\vSS$ is a diagonal singular value matrix. If we set $\vZ=\vU\vSS^\frac{1}{2}$, then the complex symmetric matrix $\vM$ enjoys symmetric factorization $\vM=\vZ\vZ^T$.

Therefore, we can employ low-rank symmetric factorization $\G\vz=\vZ\vZ^T$ to eliminate the rank constraint, where $\vZ \in\C^{n_s\times r}$, and  enforce Hankel structure by the following Hankel constraint
\begin{align*}
(\I-\G\G^*)(\vZ\vZ^T)=\bm{0},    
\end{align*}
where $\G\G^*$ is a projector that maps a matrix to a Hankel matrix.
Thus \eqref{eq:low_rank_G} can be rewritten as
\begin{align*}
&\min_{\vZ\in\C^{n_s\times r}}\la\P_{\Omega}\(\G^*\(\vZ\vZ^T\)-\vy\), \G^*\(\vZ\vZ^T\)-\vy\ra \\
&\quad\mbox{s.t.}\quad(\I-\G\G^*)(\vZ\vZ^T) = \bm{0},\numberthis\label{eq:constrained}\end{align*}
where we use the fact that $\vz=\G^*\(\G\vz\)=\G^*\(\vZ\vZ^T\)$, and substitute this in \eqref{eq:low_rank_G}.

Furthermore, we consider a penalized version of \eqref{eq:constrained} to recover the weighted desired signal $\vy$,
\begin{align*}
f(\vZ) = &\frac{1}{4p}\la\P_{\Omega}(\G^*(\vZ\vZ^T)-\vy),\G^*(\vZ\vZ^T)-\vy\ra\\
&+\frac{1}{4}\ln(\I-\G\G^*)(\vZ\vZ^T)\rn_F^2, \numberthis\label{eq:penalized}
\end{align*}
 where $p=m/n$ is the sampling ratio. We can interpret \eqref{eq:penalized} as that
 one uses a low-rank symmetrically factorized matrix $\vM=\vZ\vZ^T$ with Hankel structure (Hankel structure enforced in \eqref{eq:penalized} by a penalty term) to estimate the lifted matrix $\G\vy$ via minimizing the mismatch in the measurement domain. 

\section{Definitions and algorithms} \label{sec:def-algorithm}

\subsection{Definitions}
Denoting the lifted truth matrix as $\vM_\star=\G\vy$, {which is a complex symmetric matrix as pointed out in the previous section}.
Let the Takagi factorization of $\vM_\star$ being $\vM_\star=\vU_\star\vSS_\star\vU_\star^T$ and set 
\begin{align*}
 \vZ_\star=\vU_\star\vSS_\star^\frac{1}{2},\numberthis\label{eq:defZstar}   
\end{align*}
then $\vM_\star=\vZ_\star\vZ_\star^T$. 

If the singular vectors of $\vM_\star$ are aligned with the sampling basis, we can't recover $\vM_\star$ from element-wise observations. Consequently, we suppose that $\vM_\star=\G\vy$ is $\mu_0$-incoherent as in \cite{Cai2018,Candes2009}, 
\begin{definition}\label{def:coh}
The $n_s\times n_s$ square Hankel matrix $\vM_\star=\G\vy$ is $\mu_0$-incoherent, if there exists an absolute numerical constant $\mu_0>0$ such that   
\begin{align*}
\ln\vU_\star\rn_{2,\infty} \leq \sqrt{\frac{2\mu_0r}{n}},
\end{align*}
where $n=2n_s-1$, $\vU_\star\vSS_\star\vU_\star^T$ is the Takagi factorization of $\vM_\star$.
\end{definition}
\begin{remark}
    We remove the dependence on  $c_s$ for the incoherence condition defined in \cite{Chen2014, Cai2018,Cai2019}, as $c_s=n/n_s\leq 2$ for the square case and only suppose the incoherence condition for one single matrix $\vU_\star$.
\end{remark}
It has been demonstrated in \cite{Cai2018} and \cite[Thm. 2]{Liao2016} that $\vM_\star=\G\vy$ is $\mu_0$-incoherent as long as the minimum wrap-around distance between the frequencies is greater than about $2/n$, and the spectrally sparse signals are undamped.


Let $\mu$ and $\sigma$ be numerical constants such that $\mu\geq \mu_0$ and $\sigma\geq \sigma_1(\vM_\star)$. When $\vM_\star$ is $\mu_0$-incoherent, the matrix $\vZ_\star$  satisfies $\ln\vZ_\star\rn_{2,\infty}\leq\sqrt{2\mu r\sigma/n}$. Moreover, let $\CS$ be a convex set defined as 
\begin{align*}
\CS = \lcb\vZ\in\C^{n_s\times r}~|~\ln\vZ\rn_{2,\infty}\leq 2\sqrt{\frac{\mu r\sigma}{n}} \rcb.\numberthis\label{eq:set_C}
\end{align*}
We immediately have $\vZ_\star\in\CS$  
and it is reasonable to project the factor $\vZ$ onto {this set}. 
We define the following projection operator $\P_{\CS}(\cdot)$, for $\vZ\in\C^{n_s\times r}$
\begin{align*}
[\P_{\CS}(\vZ)]^{(i,:)}=\begin{cases}\vZ^{(i,:)} & \mbox{if }\|\vZ^{(i,:)}\|_2\leq 2\sqrt{\frac{\mu r\sigma}{n}},\\
\frac{\vZ^{(i,:)}}{\|\vZ^{(i,:)}\|_2}2\sqrt{\frac{\mu r\sigma}{n}}&\mbox{otherwise}.
\end{cases}
\end{align*}
\subsection{Algorithms}
As discussed above, we construct the following constrained optimization problem under symmetric factorization
\begin{align*}
\min_{\vZ\in\CS}~f(\vZ) =& \frac{1}{4p}\la\P_{\Omega}(\G^*(\vZ\vZ^T)-\vy),\G^*(\vZ\vZ^T)-\vy\ra\\
&+\frac{1}{4}\ln(\I-\G\G^*)(\vZ\vZ^T)\rn_F^2.\numberthis\label{eq:rec_algloss}
\end{align*}
 Under Wirtinger calculus, the gradient of the loss function $f(\vZ)$ is 
 \begin{small}
  \begin{align*}
\nabla f(\vZ)= p^{-1}\(\G\P_{\Omega}(\G^*(\vZ\vZ^T)-\vy)\)\bar{\vZ}+(\I-\G\G^*)(\vZ\vZ^T)\bar{\vZ}.
\numberthis\label{eq:loss_gradient} 
\end{align*}   
 \end{small}

Note that the gradient of our loss function avoids the computation and storage for two factors as well as a balancing regularization term, compared to PGD \cite{Cai2018}. 
\begin{algorithm}[t]
\caption{Symmetric Hankel Projected  Gradient Descent}
\label{alg:SHGD}
\begin{algorithmic} 
\Statex \textbf {Preprocessing:} 
 1. Make each dimension of the signal odd by zero-padding. 
 \Statex  2. {Partition} ${\Omega}$ into disjoint sets ${\Omega}_0,\cdots,{\Omega}_K$ of equal size $\hat{m}$, let $\hat{p}=\frac{\hat{m}}{n}$.
\Statex \textbf{Initialization:} $\vM^{0}=\T_r \(\hat{p}^{-1}\G\P_{{\Omega}_{0}}(\vy)\)=\vU^0\vSS^0(\vU^0)^T$, $\tilde{\vZ}^0=\vU^0(\vSS^0)^{1/2}$ and $\vZ^0=\PC{\tilde{\vZ}^0}$.
\For{$k=0,1,\cdots,K$}\\
\quad1. ${\tilde{\vZ}^{k+1}}=\vZ^k-\eta \nabla f^{(k)}(\vZ^k)$\\
\quad2. $\vZ^{k+1}=\PC{\tilde{\vZ}^{k+1}}$
\EndFor
\Statex \textbf{Output:} $\vZ^K$ in the last iteration, $\vy^K = \G^*(\vZ^K(\vZ^K)^T)$ and $\vx^K = \D^{-1}\vy^K$.
\end{algorithmic}
\end{algorithm}

We design a projected gradient descent algorithm with sample-splitting for solving \eqref{eq:rec_algloss}, named as Symmetric Hankel Projected Gradient Descent, seeing Algorithm \ref{alg:SHGD}.
%
In preprocessing, we first transform the dimension of the observed signal odd by zero-padding if not. Secondly, the sampling set ${\Omega}$ is partitioned into $K+1$ disjoint subsets of equal size $\hat{m}$. Such partitions of the sample set are commonly used in analyzing matrix completion problems \cite{Cherapanamjeri2017, Jain2013} and Hankel matrix completion problems \cite{Cai2019, Zhang2018, Zhang2019}. The sample-splitting technique keeps the independence between the current sampling set and the previous iterates, which simplifies the theoretical analyses.

The following is the initialization, which is one-step hard thresholding after Takagi factorization and then projecting the initial factor to the convex set $\CS$. There are several algorithms to implement Takagi factorization  \cite{Xu2008, Chebotarev2014, Ikramov2012} and one may use the algorithm in \cite{Xu2008} directly.

Finally, we enter the stage that iteratively updates on single factor $\vZ$.
Considering the sample splitting setting, the gradient term $\nabla f^{(k)}(\vZ^k)$ in the $k$-th iteration is 
\begin{align*}
\nabla f^{(k)}(\vZ^k)&= \hat{p}^{-1}\(\G\P_{{\Omega}_{k}}(\G^*(\vZ^k({\vZ^{k}})^T)-\vy)\)\bar{\vZ}^{k}\\
&\quad+(\I-\G\G^*)(\vZ^k({\vZ^k})^T)\bar{\vZ}^{k}.
\end{align*}

Our model and algorithm can be generalized to multi-dimensional spectral sparse signals. We omit the details due to space limitations. 

\subsection{Computational complexity for SHGD and PGD}
{
The empirical computational efficiency improvement of SHGD relies on specific implementations and characteristics of the problem being considered, such as the problem scale $n$ and the number of sinusoids $r$. A detailed comparison of the empirical computational time between SHGD and PGD is shown in Section \ref{subsec:comptime}. In this part, we provide an analysis of computational complexity for SHGD and PGD as follows. 

The gradient of SHGD \eqref{eq:loss_gradient} can be reformulated as $\nabla f(\vZ)=\G\( p^{-1}\P_{\Omega}(\G^*(\vZ\vZ^T)-\vy)-\G^*(\vZ\vZ^T)\)\bar{\vZ}+\vZ(\vZ^T\bar{\vZ})$. We compute
$\G^*(\vZ\vZ^T)$ by $r$ fast convolutions. Setting $\vw=p^{-1}\P_{\Omega}(\G^*(\vZ\vZ^T)-\vy)-\G^*(\vZ\vZ^T)$, $\(\G\vw\)\bar{\vZ}$ can be computed by $r$ fast Hankel matrix-vector multiplications. Each of the aforementioned two steps needs $Cnr\log(n)$ flops, where $C>0$ is a constant\footnote{The constant factors may differ for the two steps, but we can choose the larger one as $C$.}.  Besides, the computation of $\vZ(\vZ^T\bar{\vZ})$ requires $2n_s r^2(\approx nr^2)$ flops. For PGD \cite{Cai2018}, there are three steps of convolution type, each requiring $Cnr\log(n)$ flops, and four steps associated with balancing regularization, each using about $n r^2$ flops. Thus the computation time ratio between SHGD and PGD per-iteration is $\frac{2C nr\log(n)+nr^2}{3C nr\log(n)+4nr^2}=\frac{2C\log(n)+r}{3C\log(n)+4r}$. Depending on the relative scale between $r$ and $C\log(n)$, this ratio ranges from $1/4$ to $2/3$.
}

\section{Theoretical results and Analysis} \label{sec:theoretical-results}
In this section, we first discuss the challenges associated with analyzing the proposed SHGD algorithm. We then present the linear convergence result of SHGD. Lastly, we introduce our analysis framework and provide the proof of Theorem \ref{thm:recovery_guarantee}.

\subsection{Theoretical challenges}\label{subsec:challenge}

The previous analysis for PGD primarily relies on asymmetric factorization and balancing regularization \cite{Cai2018}, which makes it feasible to design a distance metric to effectively address unitary ambiguity. However, extending the analysis from PGD to SHGD is challenging. {The complex symmetric factorization is novel to the prior low-rank factorization model, introducing a new factorization ambiguity $\vM_\star=\vZ_\star\vZ_\star^T=\vZ_\star\vQ(\vZ_\star\vQ)^T$ for $\vM_\star\in\C^{n_s\times n_s}$, where $\vQ$ is called {complex orthogonal matrix (not unitary unless real)} \cite{Horn2012,Horn1999,Luk1997} such that  $\vQ\in \mathcal{Q} \triangleq \{\vS \in \C^{r\times r}| \vS\vS^T=\vS^T\vS=\vI\}$. Such factorization ambiguity under complex orthogonal transformation is first pointed out in our work.} The analysis of complex orthogonal ambiguity is more challenging than that of the unitary ambiguity, since $\|\vZ\vR\|=\|\vZ\|$ holds for all unitary matrix $\vR$ while $\|\vZ\vQ\|\neq\|\vZ\|$ can occur for some complex orthogonal matrix $\vQ$ as shown in Example \ref{ex1}. To characterize the distance from any matrix $\vZ$ to the set $\vZ_\star\vQ$ as indicated in \cite{Cai2018,Zheng2016,Tu2016,Ma2019}, it is essential to define a distance metric 
\begin{align*}
\distQ{\vZ}{\vZ_\star}=\inf_ {\vQ\in \mathcal{Q}}\ln\vZ-\vZ_\star\vQ\rn_F.\numberthis \label{eq:def_dis_complex_oth}
\end{align*}
However, it is hard to find the closed-form solution or a concise first-order optimality condition for \eqref{eq:def_dis_complex_oth}. Besides, the set of the complex orthogonal matrix $\vQ$ may be unbounded and therefore non-compact, bringing existence issues to the solution of the previous definition.  See examples as follows.
\begin{example} \label{ex1}
Define hyperbolic functions $\cosh{t}=(e^t+e^{-t})/2$ and $\sinh{t}=(e^t-e^{-t})/2$ . Let $\vI_{2}$ be a $2\times 2$ identity matrix and $\vS=\left[
 \begin{smallmatrix}
    0 &1\\-1 &0
\end{smallmatrix}\right]$, then $\vQ(t)=(\cosh{t})\vI_{2}+(j\sinh{t})\vS$ is a complex orthogonal matrix for all $t\in\R$ and the set $\vQ=\{\vQ(t)|t\in\R\}$ is non-compact. 
\begin{proof}
\begin{align*}
        \vQ(t)\vQ(t)^T=\vQ(t)^{T}\vQ(t)=\((\cosh{t})^2-(\sinh{t})^2\)\vI_{2}=\vI_{2},
\end{align*}
thus $\vQ(t)$ is a complex orthogonal matrix. Besides, it is obvious that  $\ln\vQ(t)\rn$ and $\ln\vQ(t)\rn_F$ $\rightarrow\infty$ as $t\rightarrow\infty$. So the set $\vQ$ is not bounded and therefore non-compact.
\end{proof}
\end{example}
{
Notice that $\mbox{dist}^2_Q(\vZ,\vZ_\star)$ can be reformulated as
\begin{small}
       \begin{align*}
2\mbox{dist}^2_Q(\vZ,\vZ_\star)
&=\inf_ {\substack{\vQ^T\vQ=\vQ\vQ^T=\vI,\\\vQ\in\C^{r\times r}}}\ln\vZ-\vZ_\star\vQ\rn_F^2+\|\vZ-\vZ_\star\vQ^{-T}\|_F^2,
\end{align*}
\end{small} $\vQ$ must be invertible 
as $\vQ^T\vQ=\vI$ for $\vQ\in{\C^{r\times 
r}}$. We drop the complex orthogonal constraint $\vQ^T\vQ=\vI$ but 
keep $\vQ$ invertible, and thus define a novel distance metric\footnote{Rigorously, we should use $\inf$ instead of $\min$ for the distance definition. However, we will prove the existence of a minimizer in 
Lemma~\ref{lem:existence} under some conditions that the iterates of our algorithm satisfy.} 
\begin{align*}
\distP{\vZ}{\vZ_\star}{=}\min_{\substack{\vP \in \C^{r\times r} \\ \text { invertible }}}\sqrt{\ln\vZ-\vZ_\star\vP\rn_F^2+\ln\vZ-\vZ_\star\vP^{-T}\rn_F^2}.\numberthis\label{eq:def_dist_invert}
\end{align*}
Additionally, we 
reveal the relationship between the previous distance metrics 
through the following Lemma. 
 \begin{lemma} \label{lem:relations-dist} The distance metrics $\distP{\vZ}{\vZ_\star}$ and $\distQ{\vZ}{\vZ_\star}$ satisfy the following relationships:
 
\noindent1) $\distP{\vZ}{\vZ_\star}\leq\sqrt{2}\distQ{\vZ}{\vZ_\star}$. 

\noindent2) Suppose  $\distP{\vZ}{\vZ_\star}\leq\varepsilon\sigma_r(\vM_\star)^{\frac{1}{2}}$, then  $\distP{\vZ}{\vZ_\star}$ asymptotically reduces to $\sqrt{2}\distQ{\vZ}{\vZ_\star}$ 
when $\varepsilon\rightarrow 0$.
 \end{lemma}
 \begin{proof}
     See Appendix \ref{apd:pf-lemma1-relatdist}. 
 \end{proof}
Lemma \ref{lem:relations-dist} tells us that $\distP{\vZ}{\vZ_\star}$ approximates $\distQ{\vZ}{\vZ_\star}$ well under some conditions. This inspires us to study $\distP{\vZ}{\vZ_\star}$ and build the corresponding convergence.} The distance metric $\distP{\vZ}{\vZ_\star}$ has a nice first-order optimality condition: 
\begin{align*}
 (\vZ_\star\vP_{\vZ})^H({\vZ-\vZ_\star\vP_{\vZ}})=(\vZ-\vZ_\star\vP_{\vZ}^{-T})^T\big(\overline{\vZ_\star\vP_{\vZ}^{-T}}\big), \numberthis\label{eq:optcond_primal} 
\end{align*}
where we define the optimal solution $\vP_{\vZ}$ for \eqref{eq:def_dist_invert}, if exists, as:
\begin{align*}
\vP_{\vZ}:=\underset{\substack{\vP \in \C^{r\times r} \\ \text { invertible }}}{\operatorname{argmin}}\sqrt{\ln\vZ-\vZ_\star\vP\rn_F^2+\ln\vZ-\vZ_\star\vP^{-T}\rn_F^2}.\numberthis \label{eq:def_optinvP}
\end{align*}

Establishing the convergence mechanism for SHGD  in terms of the distance metric $\distP{\vZ}{\vZ_\star}$ reveals a distinct deviation from that for PGD \cite{Cai2018}.
Unlike PGD, SHGD requires an analysis towards a different distance metric $\distP{\vZ}{\vZ_\star}$ and a different gradient direction under symmetric factorization. Also, the conditions about the local basin of attraction towards SHGD are characterized differently from PGD, 
which are small deviations from the true solution in terms of our distance metric and well-conditioness of $\vZ_\star\vP_{\vZ}$, $\vZ_\star\vP_{\vZ}^{-T}$. 
At last, we establish an inductive convergence analysis framework to verify that the iterates in Algorithm \ref{alg:SHGD} satisfy such conditions and thus enjoy linear convergence to the ground truth.
\subsection{Main results} \label{subsec:thm_linear_converg}
{We consider the sampling with replacement model as in \cite{Cai2018,Cai2019,Zhang2018,Zhang2019}. To distinguish from the sampling index set  $\Omega$ and $\Omega_k$, we denote ${\hat{\Omega}}=\{a_j|j=1,\cdots,m\}$ and ${\hat{\Omega}}_k=\{a_j|j=1,\cdots,\hat{m}\}$ as the corresponding sampling set with replacement, where the indices $a_j$ are drawn independently and uniformly from $\{0,\cdots,n-1\}$.}  


\begin{theorem}\label{thm:recovery_guarantee}
Assume $\vM_\star=\G\vy$ is $\mu_0$-incoherent.  Let $\eta=\frac{\eta^{\prime}}{\sigma_1(\vM_\star)}$, where $0<\eta^{\prime}\leq\frac{1}{54}$. Set $\sigma=\sigma_1(\vM_0)/(1-\varepsilon_0)$ and  $\mu\geq \mu_0$, where $\varepsilon_0$ is a small enough constant. If the number of observations satisfies $m\gtrsim O \(\varepsilon_0^{-2}\mu^2\kappa^4 r^2\log(n)\)$, the iterates $\vZ^k$ of Algorithm~\ref{alg:SHGD} satisfy
\begin{align*}
    \distP{\vZ^k}{\vZ_\star} \leq \left(1-\frac{11\eta^{\prime}}{100\kappa}\right)^k c_1\hspace{0.05cm}\sigma_r(\vM_\star)^{\frac{1}{2}},
\end{align*}
with probability at least $1-O(n^{-2})$, where $c_1>0$ is a universal constant, and $\kappa=\sigma_1(\vM_\star)/\sigma_r(\vM_\star)$.
\end{theorem}
\begin{remark} \label{rmk:recoveryerr_vec}
By establishing the relationship between $\|\vx^k-\vx\|_2$ and $\distP{\vZ^k}{\vZ_\star}\leq \varepsilon$, we can obtain the $\varepsilon$ recovery accuracy for $\|\vx^k-\vx\|_2$, i.e., $\|\vx^k-\vx\|_2\leq O(\varepsilon)$, when $\distP{\vZ^k}{\vZ_\star}\leq \varepsilon$. See the proof in Appendix \ref{apd:pf_rmkrecov_vec}.
\end{remark}
\begin{remark}
 Our step size is on the order of $O(\frac{1}{\sigma_1(\vM_\star)})$ independent of $r$ and $\mu$, compared to the {conservative} stepsize $O(\frac{\sigma_r(\vM_\star)}{\mu^{2}c_s^2r^2\sigma_1^2(\vM_\star)})$
 of PGD \cite{Cai2018}.
 To attach $\epsilon$ recovery accuracy, the iteration complexity is $O(\kappa\log(\frac{1}{\epsilon}))$, independent of $r$ and $\mu$ also, compared to PGD \cite{Cai2018}. 
\end{remark}

\begin{remark}
 With a more careful sample-splitting strategy\footnote{One can use a non-uniform sample-splitting strategy, which is to set the size of the sample set for the initialization satisfy $\tilde{m} \gtrsim O\( \varepsilon_0^{-2}\mu\kappa^4 r^2\log(n)\)$ from Lemma~\ref{lem:init}, and $\hat{m}\gtrsim O\( \varepsilon_0^{-2}\mu^2\kappa^2 r^2\log(n)\)$ for the following $\kappa\log(1/\varepsilon)$ steps from Lemma~\ref{lem:dist_contract}.}, the sample complexity of SHGD is $O(\varepsilon_0^{-2}\mu^2\kappa^{4}r^2\log(n)\log(1/\varepsilon))$ to attach $\varepsilon$ accuracy in terms of our distance metric. We believe that the requirement for sample-splitting is just an artifact of our proof and can be removed safely, which we leave as a potential direction for future work.  
\end{remark}
{
\begin{remark}
    By combining the techniques from Theorem \ref{thm:recovery_guarantee} and \cite{Chen2015}, we can generalize the analysis to noisy measurements $\vy_e=\P_\Omega(\vx+\ve)$ and give an order-wise recovery bound for simplicity. Suppose the elements of the noise vector $\ve$ are independent zero-mean sub-Guassian random variables with parameter $\sigma$ \cite{Vershynin2018}. With high probability, one has the following robust recovery guarantee:
 \begin{align*}
     \|\vx-\vx^k\|_2\lesssim \|\vZ^k(\vZ^k)^T-\vM_\star\|_F\lesssim\sigma\sqrt{\frac{n^2}{m}},
 \end{align*}
as long as $\sigma\ll\sigma_r(\vM_\star)$ and $m\geq O(r^2\log(n))$.
\end{remark}
}
\subsection{Analytical framework} \label{subsec:analysis}

{First of all, we introduce the existence of the minimizer of the problem \eqref{eq:def_dist_invert} under some conditions; see Lemma~\ref{lem:existence} in Appendix \ref{apd:support-lems} for the details.  One can easily verify the existence of the minimizer of the distance metric for the iterates in our Algorithm \ref{alg:SHGD} inductively by invoking Lemma~\ref{lem:existence}. Consequently, we omit the debation about such existence issues for simplicity of presentation.  }

{
We apply a two-stage analysis, combined with an inductive framework. Firstly we show some good properties of the initialization, which are about the local basin of attraction specified for SHGD. Then we show that the iterates returned by Algorithm \ref{alg:SHGD} can converge linearly to the true solution  provided they are in such a local basin of attraction. 
 }
 \begin{lemma}[Initialization conditions]\label{lem:init}
Suppose $\vM_\star=\G\vy$ is $\mu_0$-incoherent, $\sigma=\sigma_1(\vM_0)/(1-\varepsilon_0)$, $\mu\geq \mu_0$ and $\varepsilon_0$ is a small enough constant. When $\hat{m}\gtrsim O\( \varepsilon_0^{-2}\mu\kappa^4 r^2\log(n)\)$,  one has the following results with probability at least $1-n^{-2}$,
\begin{align*}
  \distP{\vZ^0}{\vZ_\star}&\leq 1.6\kappa^{-1}\varepsilon_0\sigma_r(\vM_\star)^{\frac{1}{2}};\numberthis \label{eq:initdist_afterproj}\\
 \sigma_r\(\vZ_{\star {i}}^0\)&\geq\(1-6.4\kappa^{-1}\varepsilon_0\)\sigma_r\(\vM_\star\)^{\frac{1}{2}},i=1,2,\numberthis \label{eq:init_nmbd_f1}\\
 \sigma_1\(\vZ_{\star{i}}^0\)&\leq \(1+6.4\kappa^{-1}\varepsilon_0\)\sigma_1\(\vM_\star\)^{\frac{1}{2}}, i=1,2;\numberthis \label{eq:init_nmbd_f2}\\
     \sigma&=\frac{\sigma_1(\vM^0)}{1-\varepsilon_0}\geq\sigma_1(\vM_\star).\numberthis 
\end{align*} 
 Besides, the optimal solution $\vP^0$ to $\distP{\vZ^0} {\vZ_\star}$ exists and we set $\vZ^0_{\star{1}}=\vZ_\star\vP^0,\vZ^0_{\star{2}}=\vZ_\star(\vP^0)^{-T}$.
 \end{lemma}

\begin{proof}
    See Appendix \ref{apd:proof-init-lem}.
\end{proof}


Next, we suppose the iterate in the $k$-th step satisfies the following conditions:
\begin{align*}
 &\distP{\vZ^{k}}{\vZ_\star}\leq 1.6\kappa^{-1}\varepsilon_0 \left(1-\frac{11\eta^{\prime}}{100\kappa}\right)^k\sigma_r(\vM_\star)^\frac{1}{2};\numberthis\label{eq:distancebd} \\
    &\sigma_r(\vZ_{\star{i}}^k)\geq\Big[1-6.4\kappa^{-1}\varepsilon_0\sum_{t=0}^k(1-\frac{11\eta^{\prime}}{100\kappa})^t\Big]\sigma_r\(\vM_\star\)^{\frac{1}{2}},\numberthis\label{eq:nmcond1}
    \\ 
    &\sigma_1(\vZ_{\star{i}}^k)\leq \Big[(1+6.4\kappa^{-1}\varepsilon_0\sum_{t=0}^k(1-\frac{11\eta^{\prime}}{100\kappa})^t\Big]\sigma_1\(\vM_\star\)^{\frac{1}{2}}, \numberthis\label{eq:nmcond2}
\end{align*}

\noindent where $i=1,2$. Besides, the optimal solution $\vP^k$ to $\distP{\vZ^k}{\vZ_\star}$ exists and  $\vZ^k_{\star{1}}=\vZ_\star\vP^k,\vZ^k_{\star{2}}=\vZ_\star(\vP^k)^{-T}$.

Next, we point out the following distance contraction result for the gradient update that $\tilde{\vZ}^{k+1}=\vZ^k-\eta \nabla f^{(k)}(\vZ^k)$ under conditions \eqref{eq:distancebd}, \eqref{eq:nmcond1}, and \eqref{eq:nmcond2}.
  
\begin{figure*}[!t]
		\centering
		\includegraphics[width=0.19\linewidth]{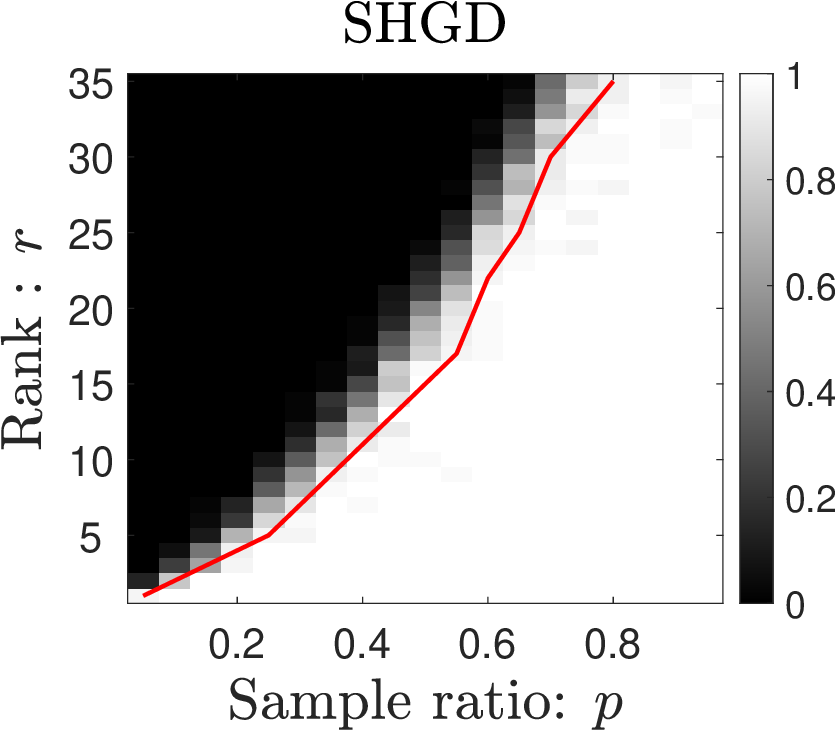}  \hfill
  \includegraphics[width=0.19\linewidth]{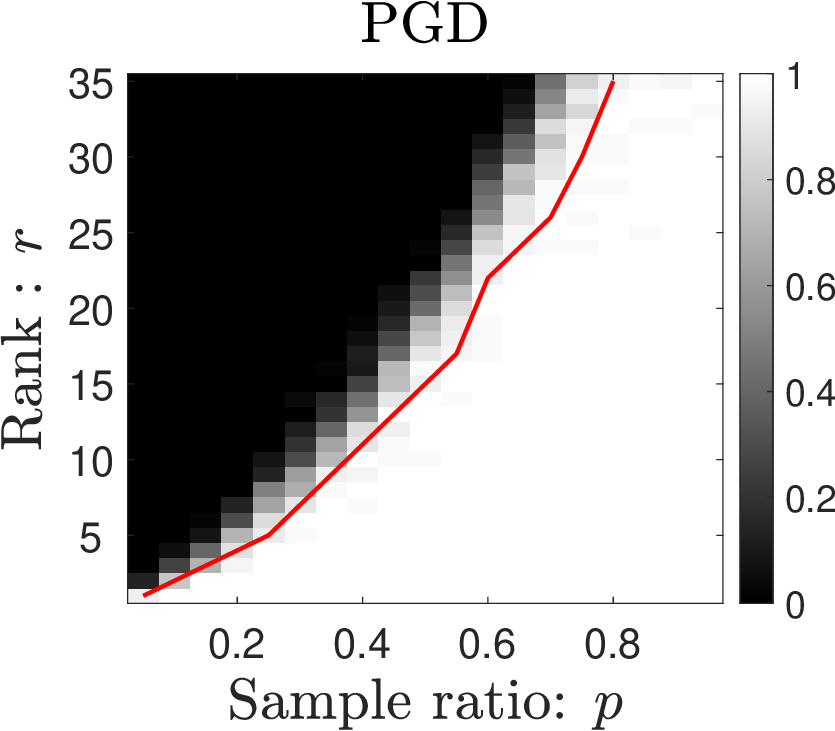}  \hfill
  \includegraphics[width=0.19\linewidth]{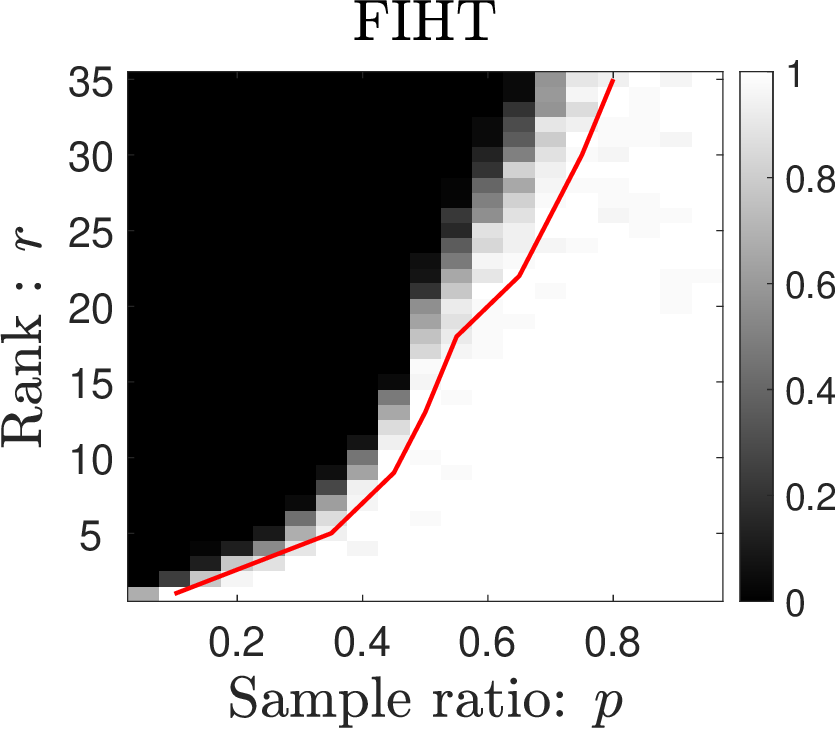} \hfill
  \includegraphics[width=0.19\linewidth]{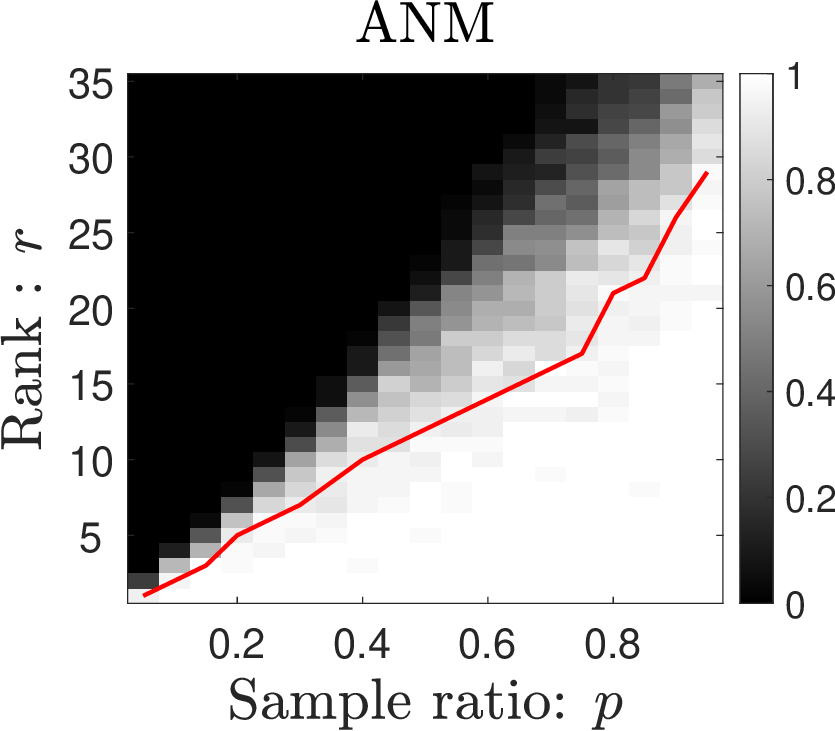}
  \hfill
  \includegraphics[width=0.19\linewidth]{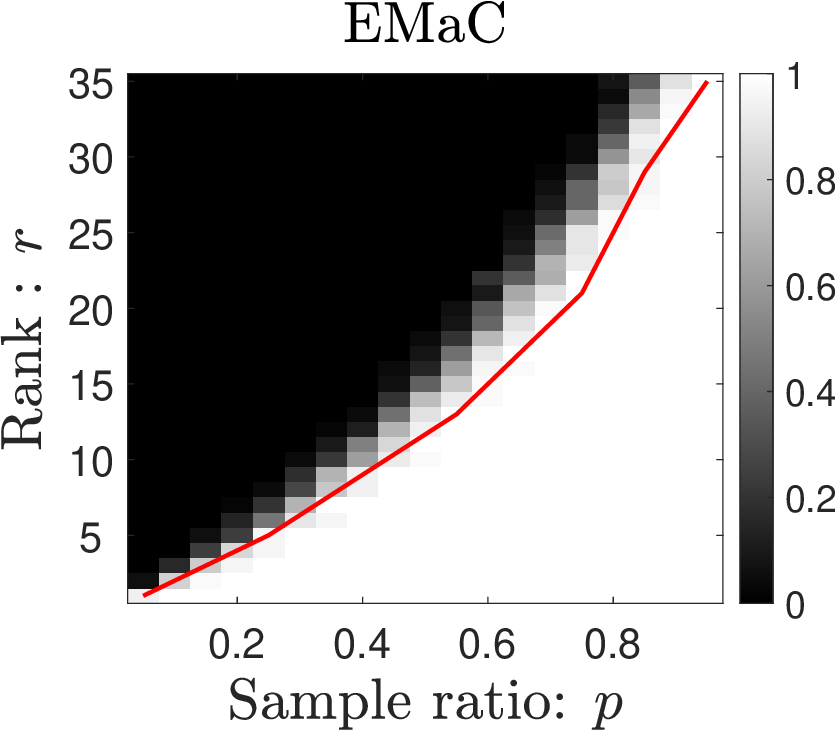} 
		\caption{The comparisons of different algorithms in terms of phase transitions under frequencies without separation.  The red curve is the 90\% success rate curve.}
		\label{fig:phasetrans_withoutsep}
  \end{figure*}

 \begin{figure*}[!t]
		\centering
		\includegraphics[width=0.19\linewidth]{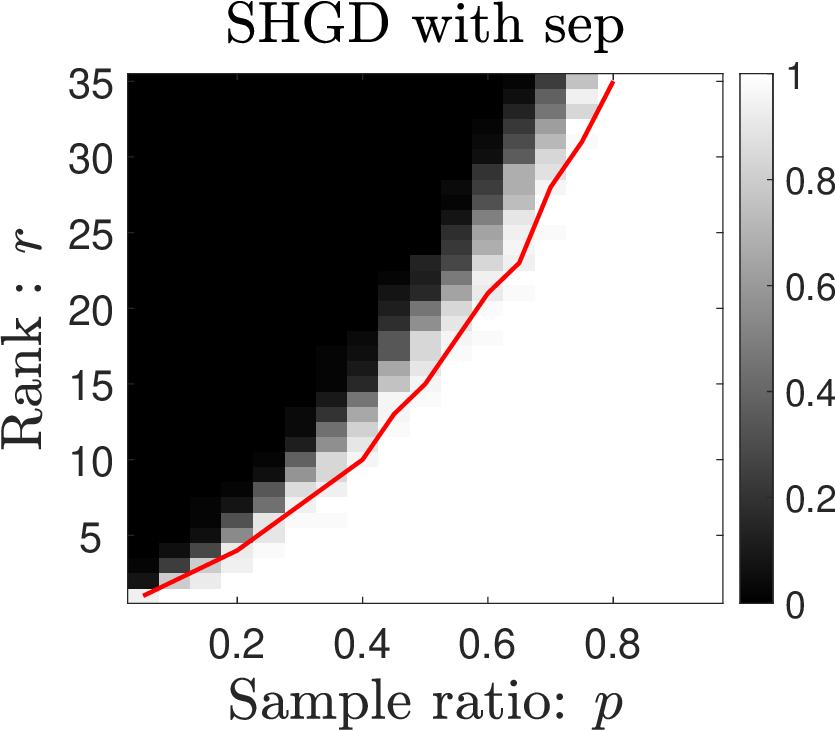}  \hfill
  \includegraphics[width=0.19\linewidth]{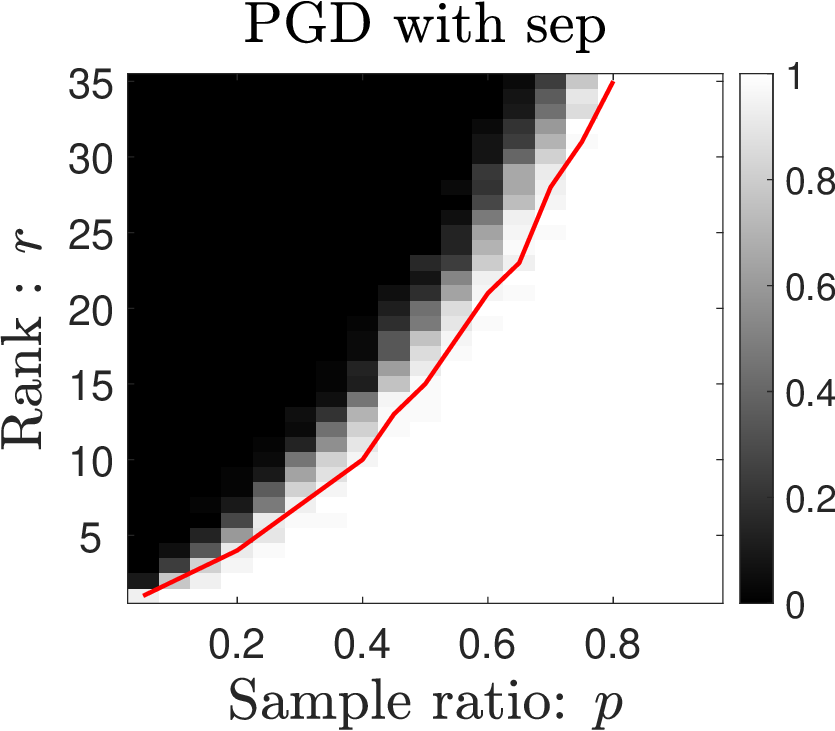}  \hfill
  \includegraphics[width=0.19\linewidth]{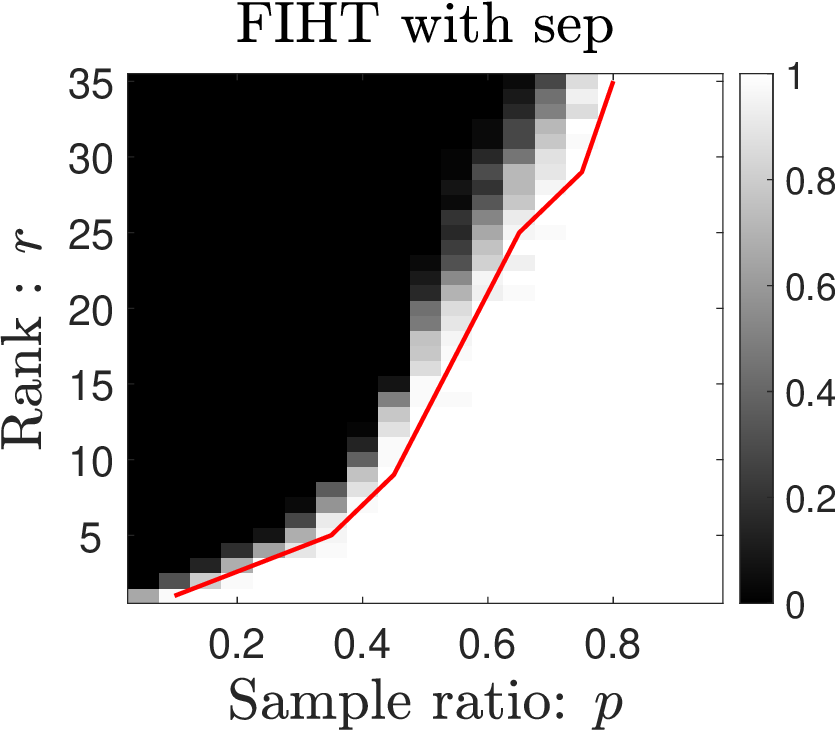} \hfill
  \includegraphics[width=0.19\linewidth]{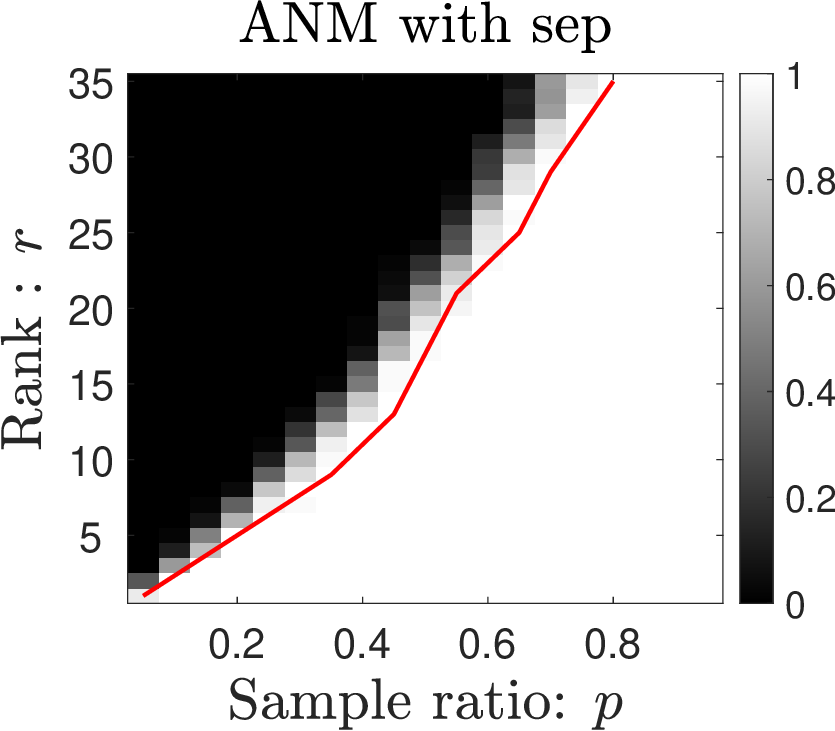} \hfill
  \includegraphics[width=0.19\linewidth]{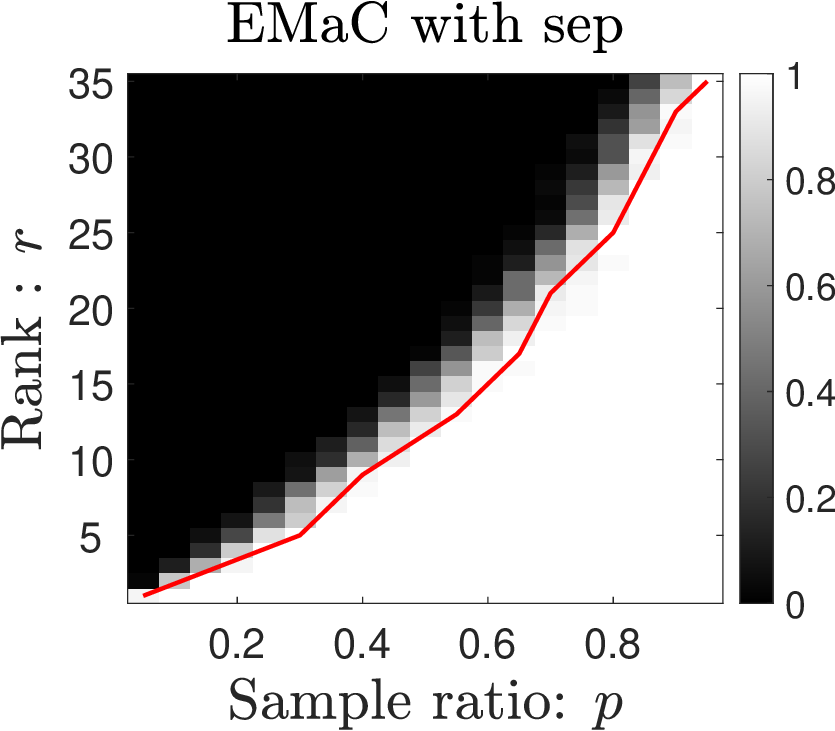} 
		\caption{The comparisons of different algorithms in terms of  phase transitions under  
 frequencies with separation. The red curve is the 90\% success rate curve.}
		\label{fig:phasetrans_withsep}
\end{figure*}

\begin{lemma}[Distance contraction]\label{lem:dist_contract}
Suppose $\vM_\star=\G\vy$ is $\mu_0$-incoherent, and the conditions of \eqref{eq:distancebd}, \eqref{eq:nmcond1}, and \eqref{eq:nmcond2} hold. Let $\eta=\frac{\eta^{\prime}}{\sigma_1(\vM_\star)}$ where $0<\eta^{\prime}\leq\frac{1}{54}$. When $\hat{m}\gtrsim O\( \varepsilon_0^{-2}\mu^2\kappa^2 r^2\log(n)\)$, the iterate $\tilde{\vZ}^{k+1}$ of Algorithm~\ref{alg:SHGD} satisfies
    \begin{align*}
 \distP{\tilde{\vZ}^{k+1}}{\vZ_\star} \leq \left(1-\frac{11\eta^{\prime}}{100\kappa}\right)\distP{\vZ^k}{\vZ_\star},\numberthis \label{eq:dist_contraction_npj}
\end{align*}
with probability at least $1-c_2 n^{-2}$, and $\varepsilon_0$ is a small enough constant.
\end{lemma}
\begin{proof}
    See Appendix \ref{apd:proof-lem2}.
\end{proof}

Lemma~\ref{lem:dist_contract} reveals that when the iterates satisfy conditions about the local basin of attraction, they enjoy a linear convergence rate to the truth after the gradient update step. {We verify the iterates returned by Algorithm \ref{alg:SHGD} satisfy such conditions via an inductive manner, the details of which are deferred to subsection \ref{subsec:thmpf}.}
\subsection{Proof of Theorem \ref{thm:recovery_guarantee}}\label{subsec:thmpf}
\begin{proof}
We establish  Theorem \ref{thm:recovery_guarantee} via an inductive approach with the conditions \eqref{eq:distancebd}, \eqref{eq:nmcond1}, and \eqref{eq:nmcond2} as the inductive hypotheses. If we can inductively establish \eqref{eq:distancebd}, we finish the proof of Theorem \ref{thm:recovery_guarantee}. The hypotheses \eqref{eq:nmcond1} and \eqref{eq:nmcond2} are necessary to establish \eqref{eq:distancebd} for the next iteration (eg: from the $k$-th to the $(k+1)$-th step). Therefore we need to establish \eqref{eq:distancebd}, \eqref{eq:nmcond1}, and \eqref{eq:nmcond2} together in an inductive manner.

First, it is evident that the hypotheses are true for $k=0$ from Lemma~\ref{lem:init}, with probability at least $1-n^{-2}$ when $\hat{m}\gtrsim O\( \varepsilon_0^{-2}\mu\kappa^4 r^2\log(n)\)$, $\sigma=\sigma_1(\vM_0)/(1-\varepsilon_0)$, $\mu\geq \mu_0$ and $\varepsilon_0$ is a small enough constant. 

Next, we provide the proof sketch from the $ k$-step to the $(k+1)$-th step; see the complete proof in Appendix \ref{apd:pf_induction}. 
The key idea is to invoke Lemma~\ref{lem:dist_contract} to establish a linear convergence mechanism and that the projection step makes the iterates contract in terms of our distance metric, {as well as keeping  well-conditioness of $\vZ_{\star{{1}}}$ and $\vZ_{\star{2}}$.} Each induction step holds with probability at least $1-c_2 n^{-2}$ when $\hat{m}\gtrsim O\( \varepsilon_0^{-2}\mu^2\kappa^2 r^2\log(n)\)$ and $\eta=\frac{\eta^{\prime}}{\sigma_1(\vM_\star)}$, where $0<\eta^{\prime}\leq\frac{1}{54}$.

combining the above results, we can inductively establish the 
 following linear convergence result of Theorem \ref{thm:recovery_guarantee}
 \begin{align*}
    \distP{\vZ^k}{\vZ_\star}\leq \left(1-\frac{11\eta^{\prime}}{100\kappa}\right)^k c_1\hspace{0.05cm}\sigma_r(\vM_\star)^{\frac{1}{2}},
\end{align*}
with  probability at least  $1-O(n^{-2})$ and $m=(k+1)\hat{m}\gtrsim O\(\varepsilon_0^{-2}\mu^2\kappa^4 r^2\log(n)\)$.  
\end{proof}

 


\section{Numerical Simulations} \label{sec:numerical}
In this section, we demonstrate the performance of SHGD via extensive numerical simulations \footnote{Our code is available at \url{https://github.com/Jinshengg/SHGD}.}. Specifically,  the signal can not be lifted as a symmetric Hankel matrix at first glance as the length of the signal is set even on purpose. The simulations are run in MATLAB R2019b on a 64-bit Windows machine with multi-core Intel CPU i9-10850K at 3.60 GHz and 16GB RAM. For each gradient updating of SHGD, we use the whole observation set rather than the disjoint subsets, as done in \cite{Zhang2019, Cai2019, Cherapanamjeri2017}.  We first compare the SHGD's performance in phase transition to nonconvex methods PGD and FIHT, and convex approaches EMaC and ANM. Then the computational efficiency of SHGD is shown in Section \ref{subsec:comptime}. We choose the square Hankel version of FIHT to compare, which further reduces the computation and storage costs. Also, the robustness of SHGD to additive noise is presented in Section \ref{subsec:robust2noise}.  Lastly, we present the application of SHGD in delay-Doppler estimation in Section \ref{subsec:delay-doppler-esti}.
\subsection{Phase transition} 	\label{subsec:phasetrans}
We compare the performance of SHGD with PGD \cite{Cai2018}, FIHT \cite{Cai2019}, EMaC \cite{Chen2014}, and ANM \cite{Tang2013} in terms of phase transition.
The frequency $f_k$ is random generated from $[0,1)$, and the amplitudes $d_k$ are selected as $d_k=(1+10^{0.5c_k})e^{-i\phi_k}$, where $c_k$ is uniformly distributed on $[0,1)$ and $\phi_k$ is uniformly sampled from $[0,2\pi)$. 

We consider 1-D signals, of which the length is $n=126$. The lifted Hankel matrix is a $63\times 64$ rectangular matrix for PGD and EMaC. The observed signal is zero-padded to be a length of the $n=127$ signal and then lifted to a $64\times64$ square Hankel matrix for SHGD and FIHT. 
ANM and EMaC are implemented using CVX. We use the backtracking line search to choose the stepsize of SHGD and PGD. SHGD, PGD and FIHT terminate when $\norm{\vx^{k+1}-\vx^k}_2/\norm{\vx^k}_2\leq10^{-7}$ or the maximum number of iterations is reached. 
We run 50 random simulations for a prescribed $(r,m)$, where $m=\lfloor pn\rfloor$ is the number of observations, sample ratio $p$ takes $19$ values from $0.05$ to $0.95$, and $r$ starts from $r=1$ and increases by $1$ until $r=35$. We run simulations for two different settings of frequencies, one is that there is no separation between $\{f_k\}_{k=1}^r$ and the other is that the separation condition $\Delta=\min_{j\neq k}\|f_j-f_k\|\geq1.5/n$ is satisfied. A test is supposed to be successful when $\norm{\vx^k-\vx}_2/\norm{\vx}_2\leq10^{-3}$.

From Fig. \ref{fig:phasetrans_withoutsep} and Fig. \ref{fig:phasetrans_withsep},  
SHGD shows almost the same performance as PGD and better performance than  EMaC in both frequencies without separation and with separation settings. Also, SHGD outperforms ANM when frequencies are generated without separation. Compared to FIHT, SHGD shows a better phase transition performance when the sampling ratio $p<0.5$, which is more typical in reality. 

 \begin{figure}[!t]
		\centering
	  \includegraphics[width=0.45\linewidth]{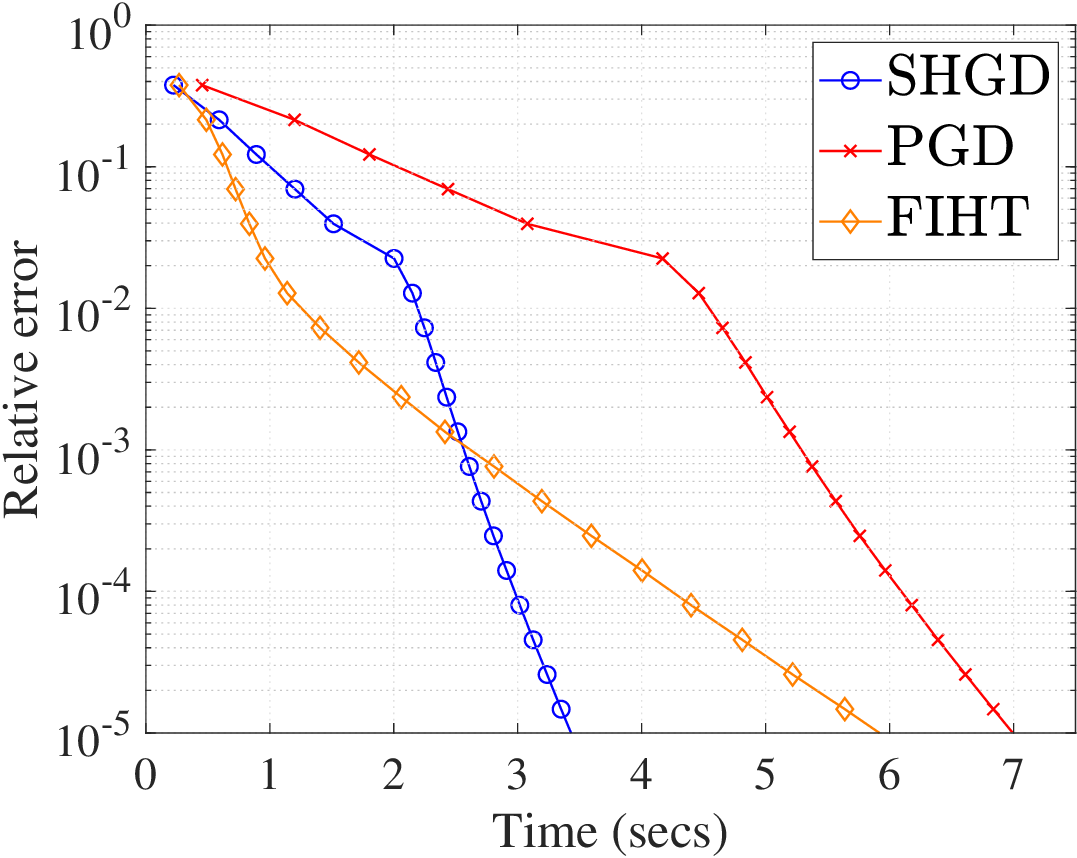}	
   \hfill
  \includegraphics[width=0.45\linewidth]{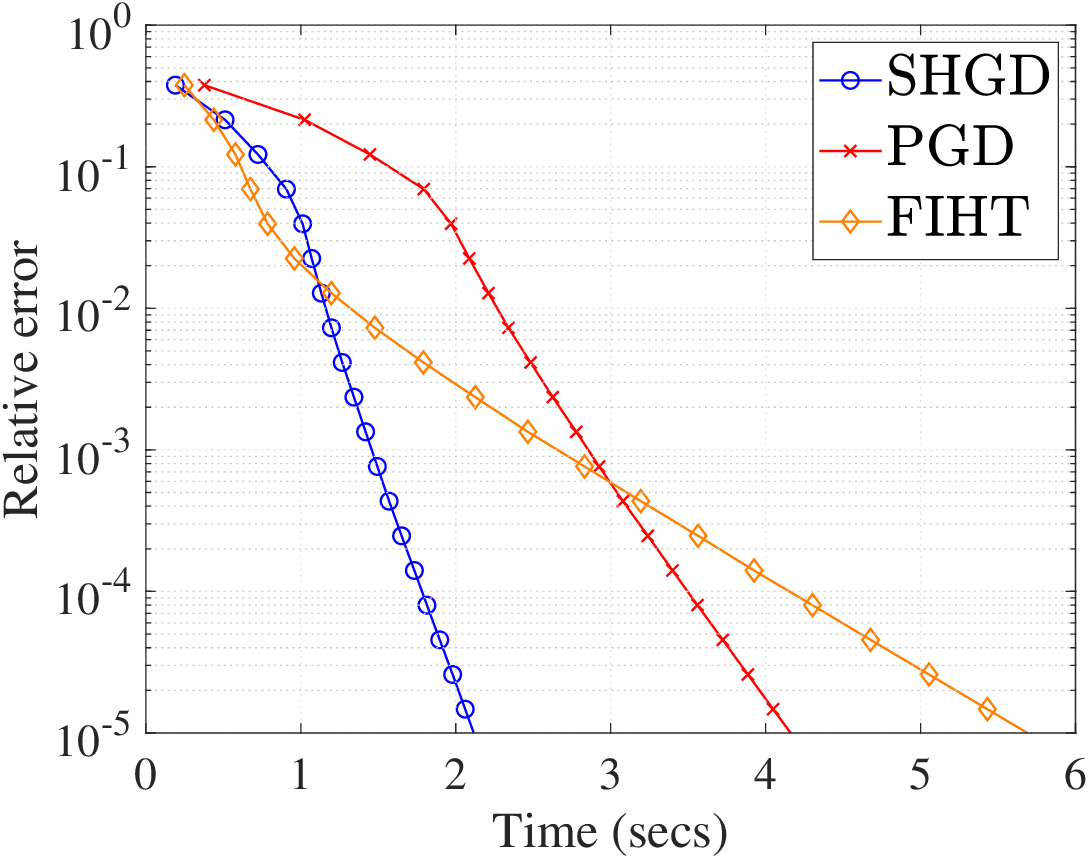} 
  \caption{The comparisons of computation time versus relative error. \textbf{Left}: Frequencies without separation. \textbf{Right}: Frequencies with separation.}
		\label{fig:timecomp_1D}
\end{figure}
 \begin{figure}[!t]
		\centering	  \includegraphics[width=0.49\linewidth]{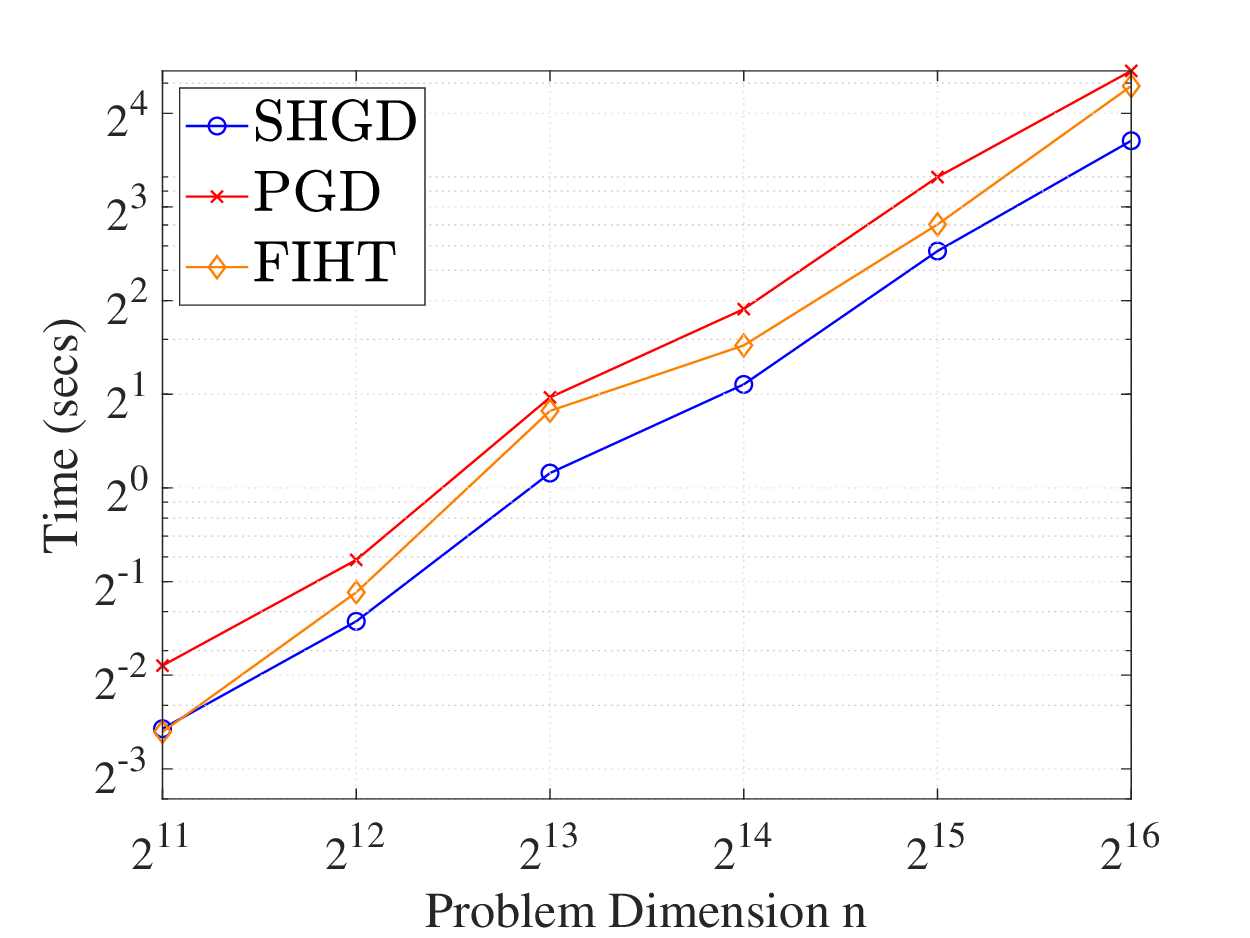}	
  \hfill
\includegraphics[width=0.49\linewidth]{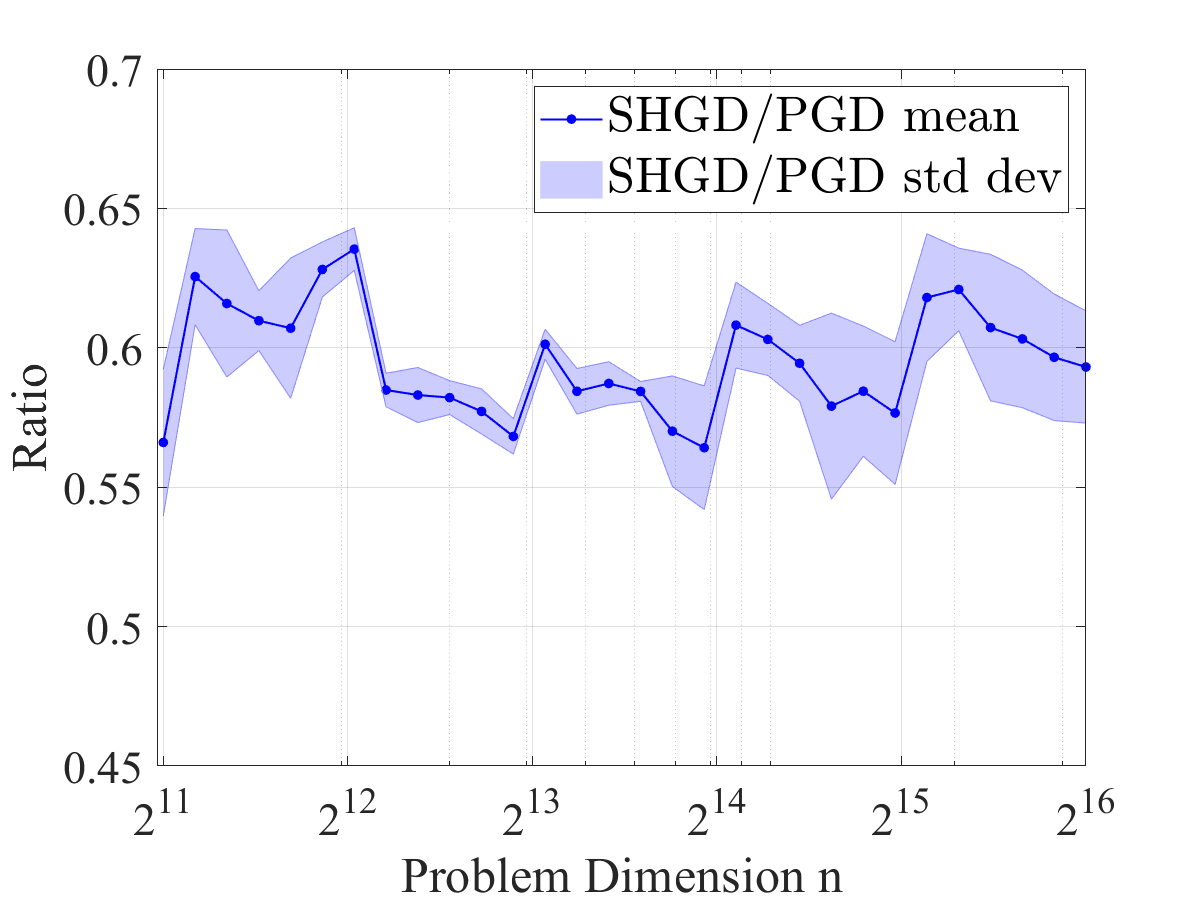} 
  		\caption{The comparisons of computation time versus different problem scales. \textbf{Left}: The average computation time of three algorithms SHGD, PGD, and FIHT. \textbf{Right}: The average computation time ratio between SHGD and PGD (SHGD/PGD).}		\label{fig:time_pbdim}
	\end{figure}  
 
\subsection{Computation time}\label{subsec:comptime}
{
This subsection conducts simulations to evaluate the computational efficiency of SHGD. In the first experiment, we examine the computation time versus the different prescribed recovery error $\norm{\vx^k-\vx}_2/\norm{\vx}_2$ and run 30 random simulations for different nonconvex fast algorithms, namely SHGD, PGD, and FIHT. 
We choose the length of the signal as $n=2046$, and the number of sinusoids as $r=150$, which is a high-dimensional situation. The number of measurements is set as 
$m=876$. After zero-padding preprocessing, the signal is lifted to 
a $1024\times 1024$ square Hankel matrix for SHGD and FIHT.  We apply a fixed stepsize strategy as $\eta=\frac{0.75}{\sigma_1(\vM_0)}$ in SHGD and PGD for faster computing.  The simulations are conducted both for frequencies with separation and without separation.  
In Fig. \ref{fig:timecomp_1D}, the average computation time of SHGD is approximately half that of PGD in both settings. 
Additionally, it can be observed that FIHT exhibits better performance in the initial stage, but becomes slower than SHGD when higher accuracy is required in both frequency settings.  

In the second experiment, we examine the computation time 
for SHGD, PGD, and FIHT  under different problem scales. Each problem scale runs 20 random simulations. The target recovery accuracy is set as $\norm{\vx^k-\vx}_2/\norm{\vx}_2=10^{-7}$. We set the number of sinusoids as $r=30$, and the frequencies  $\{f_k\}_{k=1}^r$ are generated without separations. The number of measurements is set as 
$m=512$. Firstly, we set the problem dimension as $n = 2^j-2$ with $j\in\{11,...,16\}$. The lifted matrix is a square $2^{j-1}\times2^{j-1}$ Hankel matrix for SHGD and FIHT after zero-padding preprocessing and a $(2^{j-1}-1)\times2^{j-1}$ rectangular Hankel matrix for PGD.  We examine the average computation time for SHGD, PGD, and FIHT to attain the fixed accuracy $10^{-7}$. 
Secondly, we evaluate the time ratio between SHGD and PGD. And we set the problem dimension as $n=2\lfloor2^{j-1}\rfloor-2$, where $j$ takes 30 equally spaced values from 11 to 16. From Fig. \ref{fig:time_pbdim}, we observe that SHGD behaves faster than FIHT for most problem scenarios. Besides, the computation time of SHGD is about 0.55-0.65 times that of PGD across different problem dimensions.}
{
\subsection{Robust recovery from noisy observations} \label{subsec:robust2noise}
In this subsection, we demonstrate the robustness of SHGD to additive noise. The experiments are conducted both in frequencies with separation and frequencies without separation settings.
We set the length of signals as $n=127$ and 
the number of sinusoids as $r=12$. 
 The observations are perturbed by the noise vector  $\ve=\sigma_e\cdot\|\P_{\Omega}(\vx)\|_2\cdot\vw/\|\vw\|_2$ where $\vx$ is referred to the desired signal, $\sigma_e$ be the noise level, 
and the elements of $\vw$ are i.i.d. standard complex Gaussian random variables.  SHGD is terminated when $\norm{\vx^{k+1}-\vx^k}_2/\norm{\vx^k}_2\leq10^{-7}$. 
The number of measurements is set as $m=60$ and $m=120$ respectively. The noise level $\sigma_e$ varies from $10^{-3}$ to 1, corresponding to signal-to-noise ratios (SNR) from 60 to 0 dB.  For each noise level, we conduct 20 random tests and record the average RMSE of the reconstructed signals versus SNR.  Fig. \ref{fig:stable_rec_SNR} shows that the relative reconstruction error decreases linearly as SNR increases (the noise level decreases). Besides, the relative reconstruction error reduces when the number of observations increases. 
\begin{figure}[!t]
		\centering
		\includegraphics[width=0.49\linewidth]{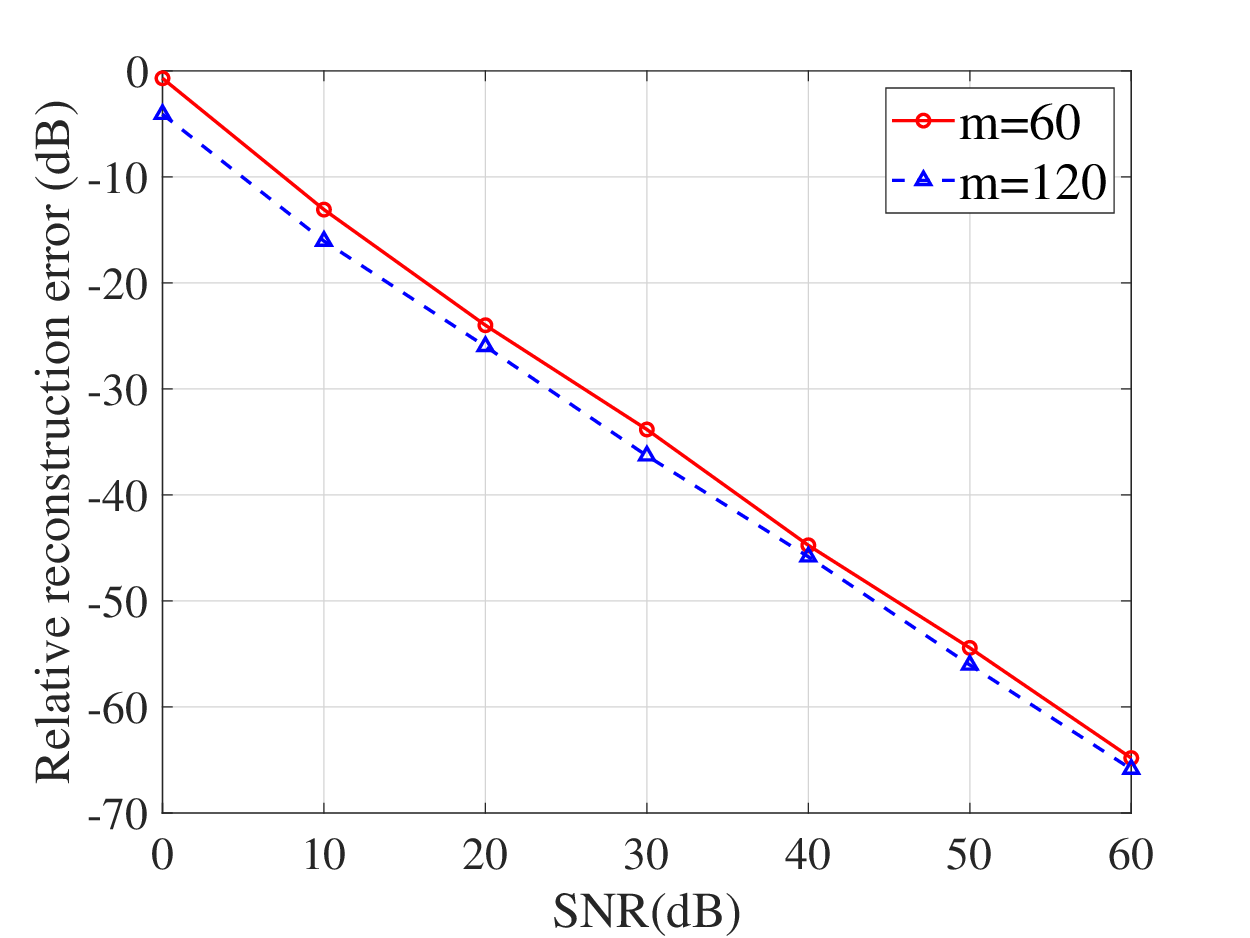} 
  \hfill
  \includegraphics[width=0.49\linewidth]{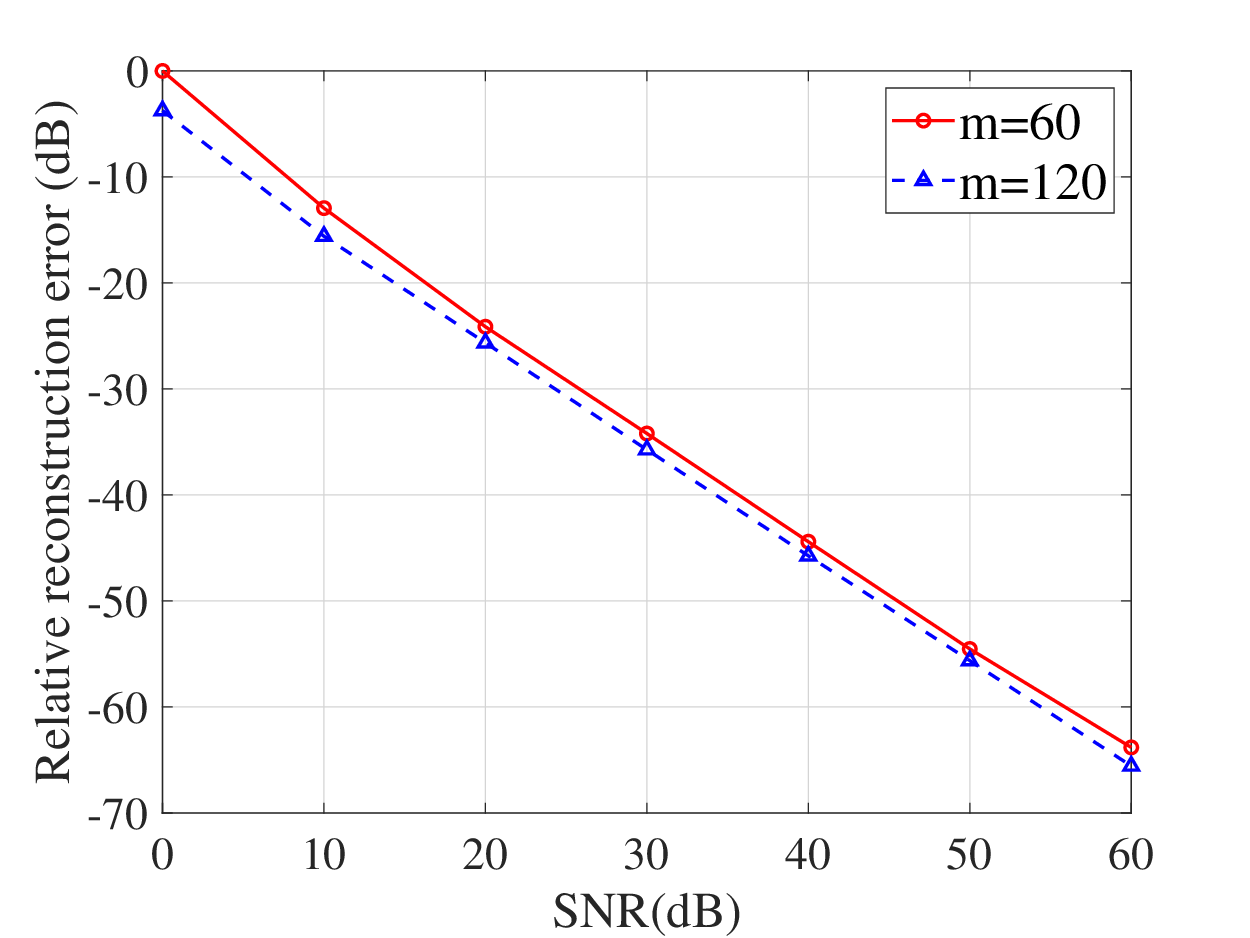}
  		\caption{The stable recovery of SHGD under different noise levels. \textbf{Left:} Frequencies without separation. \textbf{Right:} Frequencies with separation.}
		\label{fig:stable_rec_SNR}
\end{figure}
}

{\subsection{SHGD for delay-Doppler estimation based on OFDM signals}\label{subsec:delay-doppler-esti}
We consider the problem of delay-Doppler estimation from OFDM signals \cite{Sanson2020}. The channel matrix  $\vD_R$ can be obtained when the reference symbols are known, which is defined as
\begin{align*}
    \vD_R(l_1,l_2)=\sum_{k=1}^r d_k e^{i2\pi(l_2Tf_k-l_1\Delta f\tau_k)}, (l_1,l_2)\in[N_1] \times [N_2],
\end{align*}
where $N_1$ denotes the number of orthogonal subcarriers, $N_2$ denotes the number of OFDM symbols, $\{\tau_k,f_k\}_{k=1}^r$ are the delays and Doppler frequencies, $\{d_k\}_{k=1}^r$ are channel coefficients, and $T$ is the symbol duration. For simplicity of notation, we define $\psi_k=Tf_k\in[0,1)$, $\phi_k=\Delta f\tau_k\in[0,1)$.
\begin{figure}[!t]
		\centering		\includegraphics[width=0.48\linewidth]{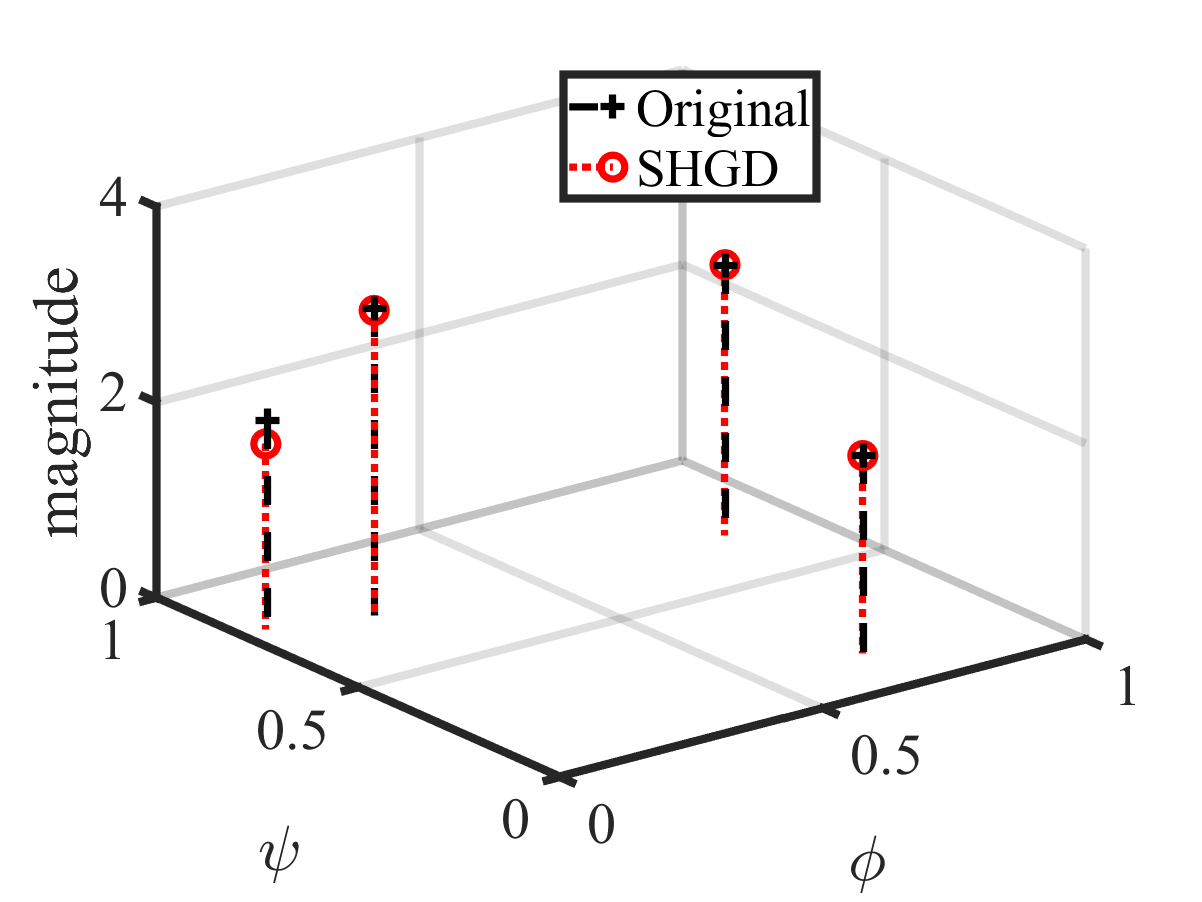} 
  \hfill  \includegraphics[width=0.48\linewidth]{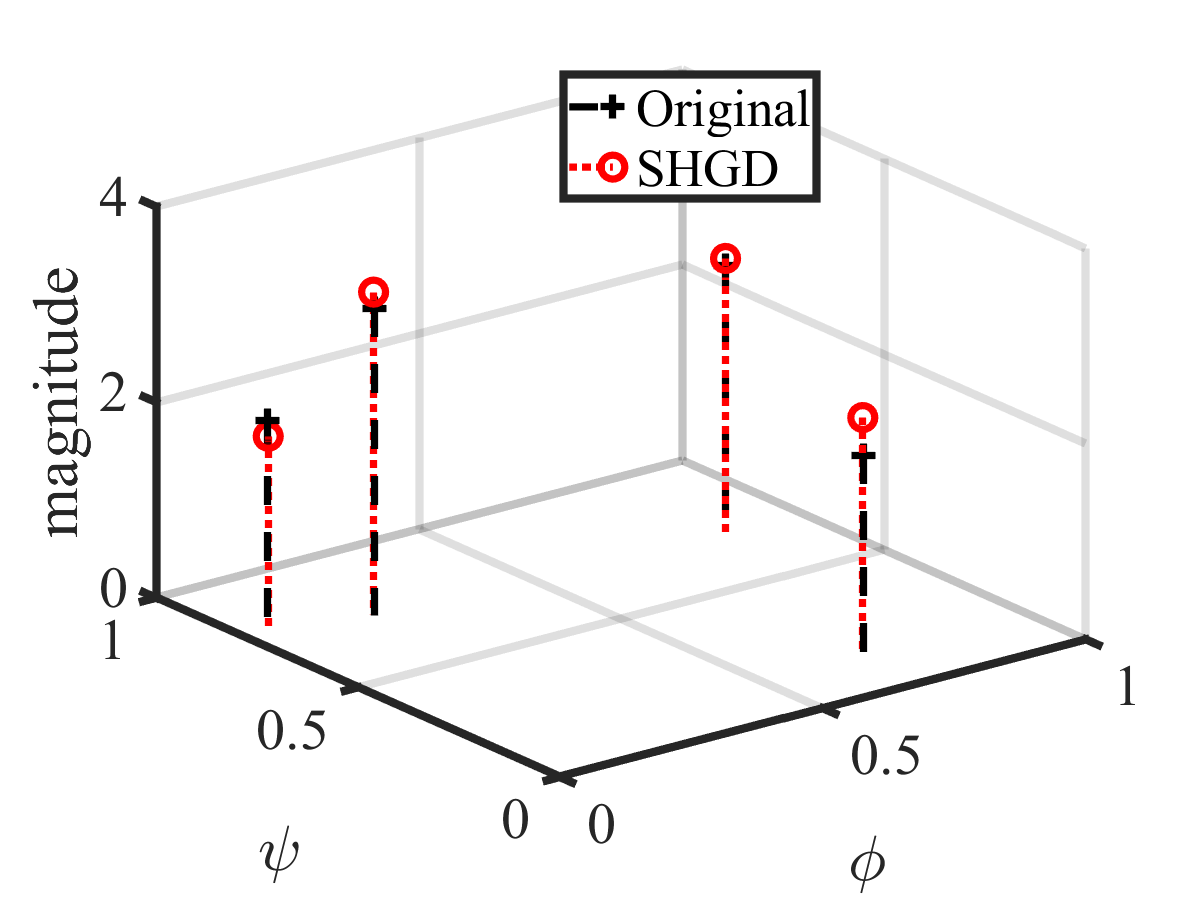}
  		\caption{ SHGD for delay-Doppler estimation from  OFDM signals. \textbf{Left}: SHGD as a denoiser when subcarriers of OFDM are equally spaced.  \textbf{Right}: SHGD as a recovery algorithm when the subcarriers of OFDM are unequally spaced.}
		\label{fig:stable_rec_visual_2D}
	\end{figure}
 
 SHGD can be viewed as a denoiser when subcarriers of OFDM are equally spaced as this corresponds to the full sampling of the channel matrix $\vD_R$.  When subcarriers of OFDM are unequally spaced as \cite{Nikookar}, it can be interpreted as the partially non-uniform row sampling of channel matrix $\vD_R$, and we use SHGD to recover the channel matrix. The strengths of non-uniform subcarriers are lower complexity of the system design and less power of intercarrier interference (ICI) \cite{Nikookar,Yang2017}. We conduct experiments on delay-Doppler estimations from standard OFDM signals and OFDM signals of non-uniform subcarriers. After executing SHGD, we use 2-D MUSIC \cite{Berger2010} to estimate delay-Doppler parameters. 
In simulations, we set $N_1=30$, $N_2=14$, and the number of targets $r=4$. We suppose the channel matrix is corrupted by the noise  $\vE=\sigma_e\cdot\|\vD_R\|_F\cdot\vW/\|\vW\|_F$ where the entries of $\vW$ are i.i.d standard complex Gaussian random variables and $\sigma_e=0.15$ is the noise level. The number of non-uniform subcarriers is set as $m=8$ and thus the number of observations of $\vD_R$ is $M=mN_2=112$. 
The 2-D channel matrix $\vD_R$ is lifted to a $128\times 128$ symmetric block-Hankel matrix after zero-padding preprocessing. Fig. \ref{fig:stable_rec_visual_2D} shows that our algorithm SHGD estimates delays, Doppler frequencies almost exactly, and channel coefficients with minor deviations. 
}
\section{Conclusions} \label{sec:conclusions}
In this paper, a new non-convex gradient method called SHGD based on symmetric factorization is proposed for the spectrally sparse signal recovery problem. Theoretical analysis guarantees that if the number of random samples is larger than $O(r^2\log(n))$, a linear convergence guarantee towards the true spectrally sparse signal can be established with high probability. The proposed SHGD algorithm reduces almost half of the cost of computation and storage compared to PGD and shows comparable computation efficiency as FIHT in numerical simulations. It is interesting to extend our symmetric factorization approach to robust spectrally sparse signal recovery and explore vanilla gradient descent via symmetric factorization when the number of sinusoids is unknown. 

\appendices  
\section{Supporting lemmas}\label{apd:support-lems}
\begin{lemma}[Existence of minimizer of the distance metric] \label{lem:existence}
  Suppose there exists an invertible matrix $\hat{\vP}$, and set $\hat{\vZ}_{\star{1}}=\vZ_{\star}\hat{\vP}$, $\hat{\vZ}_{\star{2}}=\vZ_{\star}\hat{\vP}^{-T}$. When
 \begin{align*}
     \sqrt{\|\vZ-\hat{\vZ}_{\star{1}}\|_F^2+\|\vZ-\hat{\vZ}_{\star{2}}\|_F^2} \leq \varepsilon\sigma_r(\vM_\star)^{\frac{1}{2}},
 \end{align*} 
where $\varepsilon\leq \sqrt{2}/8$ and $
\sigma_r(\hat{\vZ}_{\star{1}}) ~\mathrm{or}~\sigma_r(\hat{\vZ}_{\star{2}})\geq \frac{\sqrt{2}}{2}\sigma_r(\vM_\star)^{\frac{1}{2}}
  $, the minimizer of the problem \eqref{eq:def_dist_invert} exists.
 \end{lemma}
 \begin{proof}
   See Appendix \ref{apd:subsec_pf_existence}.
 \end{proof}
\begin{lemma}\label{lem:symdist_inq}
For symmetric $\vM_\star,\vM\in\C^{n_s\times n_s}$, their Takagi factorizations are given as  $\vM_\star=\vU_\star\vSS_\star\vU_\star^T$ and $\vM=\vU\vSS\vU^T$. Let $\vZ=\vU\vSS^\frac{1}{2}$, $\vZ_\star=\vU_\star\vSS_\star^\frac{1}{2}$, then 
\begin{align*}
\distPsq{\vZ}{\vZ_\star}\leq2\underset{\vQ \in \mathcal{O}}{\operatorname{min}}\|\vZ-\vZ_\star\vQ\|_F^2\leq\frac{\sqrt{2}+1}{\sigma_r(\vM_\star)}\|\vM-\vM_\star\|_F^2,\numberthis\label{eq:Otho_disterr}
\end{align*}
where $\mathcal{O} \triangleq \{\vS \in \R^{r\times r}| \vS\vS^T=\vS^T\vS=\vI\}$.
\end{lemma}
\begin{proof}
    See Appendix \ref{apd:proof-lem4-syminq}.
\end{proof}


We remove the dependence on $c_s$ in Lemma~\ref{lem:spec_init} and Lemma~\ref{lem:tangent_completion_inq} as $c_s\leq 2$ in the square case. 
 \begin{lemma}[{\cite[Lemma~2]{Cai2019}}]\label{lem:spec_init}
Assume $\G{\vy}$ is $\mu_0$-incoherent. Set $\bm{L}=\T_r\G(p^{-1}\P_{{\hat{\Omega}}}{(\vy)})$, and there exists a universal constant $c_3>0$ such that
$$
\|\bm{L}-\G\vy\|_2\leq c_3\sqrt{\frac{\mu_0r\log(n)}{m}}\|\G{\vy}\|_2
$$
holds with probability at least $1-n^{-2}$.
\end{lemma} 
\begin{lemma}[{\cite[Lemma~3]{Chen2014}}]\label{lem:tangent_completion_inq}
Assume $\G{\vy}$ is $\mu_0$-incoherent, and let $T$ be the tangent space of the rank $r$ matrix manifold at $\G\vy$. Then
\begin{equation*}
\|\P_{T}\G(\I-p^{-1}\P_{{\hat{\Omega}}})\G^*\P_{T}\|_2\leq\sqrt{\frac{64\mu_0r\log(n)}{m}}
\end{equation*}
holds with probability at least $1-n^{-2}$.
\end{lemma}
\begin{remark}
Let $\G\vy=\vU\vSS\vV^H$ being the SVD of an $n_s\times n_s$ matrix $\G\vy$. The subspace $T$ is denoted as
\begin{align*}
T=\{\vU\vC^H+\vD\vV^H~|~\vC\in\C^{n_s\times r},~\vD\in\C^{n_s\times r}\}.
\end{align*}
The Takagi factorization of $\G\vy$ is $\G\vy=\vU_\star\vSS_\star\vU_\star^T$ where $\vSS=\vSS_\star$, which is a special form of SVD. Therefore the  subspace  $T$ can be reformulated as
$T=\{\vU_\star\vC^T+\vD\vU_\star^T~|~\vC\in\C^{n_s\times r},~\vD\in\C^{n_s\times r}\}$.
\end{remark}

\begin{lemma}\label{lem:Hankcom_tr_inq}
    For any fixed matrices $\vA\in\C^{n_s\times r}$, $\vB\in\C^{n_s\times r}$, $\vC\in\C^{n_s\times r}$, $\vD\in\C^{n_s \times r}$
    \begin{align*}
     &\left|\la\G(p^{-1}\P_{\hat{\Omega}}-\I)\G^*(\vA\vB^T),\vC\vD^T\ra\right|\\
     &\leq c_4\sqrt{\frac{n^2\log(n)}{m}}\min\{\ln\vA\rn_F\ln\vB\rn_{2,\infty},\ln\vA\rn_{2,\infty}\ln\vB\rn_F\}\\&
     \quad \cdot\min\{\ln\vC\rn_F\ln\vD\rn_{2,\infty},\ln\vC\rn_{2,\infty}\ln\vD\rn_F\},
    \end{align*}
holds with probability at least $1-n^{-2}$ when $m\gtrsim O(\log(n))$.
\begin{proof}
    See Appendix \ref{apd:subsec_hankelcomp_tr_inq}. 
\end{proof}
\end{lemma}
\section{Proof of Lemma~\ref{lem:init}}\label{apd:proof-init-lem}
The proof is partly inspired by \cite{Ma2021}, and we make adaptations to our (complex) symmetric Hankel matrix completion problem. 

By setting $\bm{L}=\vM^0$ in Lemma~\ref{lem:spec_init}, one has
\begin{small}
\begin{align*}
\ln\vM^{0}-\G\vy\rn_F\leq\sqrt{2r}\ln\vM^{0}-\G\vy\rn_2
&\leq c_3\sqrt{\frac{2\mu_0r^2\log(n)}{\hat{m}}}\sigma_1(\vM_\star)\\
&\leq\kappa^{-1}\varepsilon_0\sigma_r(\vM_\star),\numberthis\label{eq:spectral-initial}
\end{align*}
\end{small}with probability at least $1-n^{-2}$. The last inequality follows from the assumption on $\hat{m}\gtrsim O\( \varepsilon_0^{-2}\mu\kappa^4 r^2\log(n)\)$.

As $\vM_\star=\G\vy$, a simple corollary of \eqref{eq:spectral-initial} by Weyl's Theorem is that
\begin{align*}
    \sigma_1(\vM^0)\geq (1-\kappa^{-1}\varepsilon_0)\sigma_1(\vM_\star)\geq(1-\varepsilon_0)\sigma_1(\vM_\star).
\end{align*}
Consequently $\sigma=\frac{\sigma_1(\vM^0)}{1-\varepsilon_0}\geq\sigma_1(\vM_\star)$ with high probability.

In Lemma~\ref{lem:symdist_inq}, we set $\vM=\vM^0=\vU^0\vSS^0(\vU^0)^T$, then $\vZ=\tilde{\vZ}^0=\vU^0(\vSS^0)^{\frac{1}{2}}$. combining Lemma~\ref{lem:symdist_inq} and \eqref{eq:spectral-initial} to obtain 
\begin{align*}
 \distP{\tilde{\vZ}^0}{\vZ_\star}&\leq
\sqrt{2}\min_{\substack{\vQ \in \mathcal{O}}}\|\tilde{\vZ}^{0}-\vZ_{\star}\vQ\|_F \\
&\leq\frac{(\sqrt{2}+1)^{\frac{1}{2}}}{\sigma_r(\vM_\star)^{\frac{1}{2}}}\|\vM^{0}-\G\vy\|_F\leq\frac{8\varepsilon_0}{5\kappa}\sigma_r(\vM_\star)^{\frac{1}{2}}.  
\numberthis
\label{eq:init-spec-distanceinq}
\end{align*}

Denote the optimal solution  for problem $\min_{\vQ \in \mathcal{O}}\|\tilde{\vZ}^0-\vZ_\star\vQ\|_F$ as $\vQ_z$ and we have the following result from \eqref{eq:init-spec-distanceinq}
\begin{align*}
    \|\tilde{\vZ}^{0}-\vZ_\star\vQ_{z}\|\leq 1.6\kappa^{-1}\varepsilon_0\sigma_r(\vM_\star)^{\frac{1}{2}}. 
\end{align*}

Besides, denote $\tilde{\vZ}^0_{\star{1}}=\vZ_\star\tilde{\vP}^0,\tilde{\vZ}^0_{\star{2}}=\vZ_\star(\tilde{\vP}^0)^{-T}$ where $\tilde{\vP}^0$ is the optimal solution to $\distP{\tilde{\vZ}^0}{\vZ_\star}$, as defined in \eqref{eq:def_optinvP}. The existence of $\tilde{\vP}^0$ can be verified from Lemma~\ref{lem:existence} by setting $\hat{\vP}=\vQ_z$. Thus
\begin{align*}
    \distP{\tilde{\vZ}^0}{\vZ_\star}=\sqrt{\|\tilde{\vZ}^0-\tilde{\vZ}^0_{\star{1}}\|_F^2+\|\tilde{\vZ}^0-\tilde{\vZ}^0_{\star{2}}\|_F^2},
\end{align*}
and 
\begin{align*}
 \|\tilde{\vZ}^{0}-\tilde{\vZ}^{0}_{\star{1}}\|\leq\distP{\tilde{\vZ}^0}{\vZ_\star}\leq 1.6\kappa^{-1}\varepsilon_0\sigma_r(\vM_\star)^{\frac{1}{2}}.
\end{align*}

Invoking triangle inequality, we have
\begin{align*}
\sigma_1\(\tilde{\vZ}^0_{\star{1}}\)&\leq\|\vZ_\star\vQ_z\|+\|\tilde{\vZ}^{0}-\vZ_\star\vQ_{z}\|+\|\tilde{\vZ}^{0}_{\star{1}}-\tilde{\vZ}^{0}\|\\
    &\leq(1+3.2\kappa^{-1}\varepsilon_0)\sigma_1(\vM_\star)^{\frac{1}{2}}.\numberthis\label{eq:init_inter_sig1}
\end{align*}
Additionally, from Weyl's theorem \begin{align*}
    \sigma_r\(\tilde{\vZ}_{\star{1}}^{0}\)&\geq\sigma_r\(\vZ_\star\vQ_{z}\)-\|\tilde{\vZ}_{\star{1}}^{0}-\vZ_\star\vQ_{z}\|\\
    &\geq\sigma_r\(\vZ_\star\)-\Big(\|\tilde{\vZ}^0_{\star{1}}-\tilde{\vZ}^0\|+\|\tilde{\vZ}^0-\vZ_\star\vQ_{z}\|\Big)\\
    &\geq\(1-3.2\kappa^{-1}\varepsilon_0\)\sigma_r\(\vM_\star\)^{\frac{1}{2}}.\numberthis\label{eq:init_inter_sigr}
\end{align*}
Following a similar route, we can establish the same bounds for $\sigma_1\(\tilde{\vZ}^0_{\star{2}}\)$ and $\sigma_r\(\tilde{\vZ}_{\star{1}}^{0}\)$. 
Then we have 
\begin{align*}
\|\tilde{\vZ}^0_{\star{1}}\|_{2,\infty}&\leq\ln\vU_\star\rn_{2,\infty}\sigma_1\(\tilde{\vZ}^0_{\star{1}}\)\leq 2\sqrt{\frac{\mu r\sigma}{n}},
\end{align*} 
where we invoke \eqref{eq:init_inter_sig1} and the fact that $\varepsilon_0$ is a small constant such that 
 $(1+3.2\kappa^{-1}\varepsilon_0)\leq\sqrt{2}$. The same holds for $\|\tilde{\vZ}^0_{\star{2}}\|_{2,\infty}$.
 
 Thus it is evidently that $\tilde{\vZ}^0_{\star{1}},\tilde{\vZ}^0_{\star{2}}\in\CS$ where $\CS$ is a convex set defined in \eqref{eq:set_C}. Considering $\vZ^0=\PC{\tilde{\vZ}^0}$ and the definition of $\distP{\vZ^0}{\vZ_\star}$, we immediately have
 \begin{small} 
\begin{align*}
     \distP{\vZ^0}{\vZ_\star}&\leq\sqrt{\|\vZ^0-\tilde{\vZ}^0_{\star{1}}\|_F^2+\|\vZ^0-\tilde{\vZ}^0_{\star{2}}\|_F^2}
    \\&\leq\sqrt{\|\tilde{\vZ}^0-\tilde{\vZ}^0_{\star{1}}\|_F^2+\|\tilde{\vZ}^0-\tilde{\vZ}^0_{\star{2}}\|_F^2}
    \\&\leq 1.6\kappa^{-1}\varepsilon_0\sigma_r(\vM_\star)^{\frac{1}{2}}. \numberthis \label{eq:init_opt_diff_inter}
\end{align*}
 \end{small}
The existence of ${\vP}^0$ can be verified from Lemma~\ref{lem:existence} by setting $\hat{\vP}=\tilde{\vP}^0$.
Two immediate results of \eqref{eq:init_opt_diff_inter} are
\begin{align}
    \|\vZ_{\star{1}}^{0}-\vZ^{0}\|\leq1.6\kappa^{-1}\varepsilon_0\sigma_r(\vM_\star)^{\frac{1}{2}},\label{eq:init_opt_diff3}\\
    \|\vZ^{0}-\tilde{\vZ}_{\star{1}}^{0}\|\leq 1.6\kappa^{-1}\varepsilon_0\sigma_r(\vM_\star)^{\frac{1}{2}}. \label{eq:init_opt_diff4}
\end{align}
We are still left to prove \eqref{eq:init_nmbd_f1} and \eqref{eq:init_nmbd_f2}, which are the bounds for $\vZ_{\star{1}}^{0}$ and $\vZ_{\star{2}}^{0}$. We only give the proof of bounding $\vZ_{\star{1}}^{0}$, and  the same holds for bounding $\vZ_{\star{2}}^{0}$. 

Invoking triangle inequality and combining \eqref{eq:init_inter_sig1}, \eqref{eq:init_inter_sigr}, \eqref{eq:init_opt_diff3}, \eqref{eq:init_opt_diff4}, we have
\begin{align*}
 \sigma_1\(\vZ_{\star{1}}^{0}\) &\leq \|\vZ_{\star{1}}^{0}-\vZ^{0}\|+\|\vZ^{0}-\tilde{\vZ}_{\star{1}}^{0}\|+\|\tilde{\vZ}^0_{\star{1}}\|\\
 &\leq\(1+6.4\kappa^{-1}\varepsilon_0\)\sigma_1\(\vM_\star\)^{\frac{1}{2}},
\end{align*}
 and additionally invoking Weyl's theorem
\begin{align*}
    \sigma_r\(\vZ_{\star{1}}^{0}\)&\geq \sigma_r(\tilde{\vZ}_{\star{1}}^{0})-\|\vZ_{\star{1}}^{0}-\tilde{\vZ}_{\star{1}}^{0}\|\\
    &\geq\sigma_r(\tilde{\vZ}_{\star{1}}^{0})-\Big(\|\vZ_{\star{1}}^{0}-\vZ^{0}\|+\|\vZ^{0}-\tilde{\vZ}_{\star{1}}^{0}\|\Big)\\
    &\geq\(1-6.4\kappa^{-1}\varepsilon_0\)\sigma_r\(\vM_\star\)^{\frac{1}{2}}.
\end{align*}

\section{Proof of Lemma~\ref{lem:dist_contract}}\label{apd:proof-lem2}
{As shown in \cite{Zhang2018,Zhang2019,Recht2011}, it is sufficient to study the sampling model with replacement. In the following derivations, ${\hat{\Omega}}_{k}$} is a sampling set with replacement corresponding to the index set $\Omega_{k}$ in the $k$-th iteration.
The following facts are immediate results of \eqref{eq:nmcond1} and \eqref{eq:nmcond2}, which are helpful for establishing \eqref{eq:dist_contraction_npj}.
\begin{claim} \label{claim:regcond_for_factors}
Under the conditions of \eqref{eq:nmcond1} and \eqref{eq:nmcond2}, 
\begin{align}
    \frac{1}{2}\sigma_r(\vM_\star)\leq\sigma_r^2\(\vZ_{\star{i}}^k\)&\leq\sigma_1^2\(\vZ_{\star{i}}^k\)\leq 2\sigma_1(\vM_\star).\label{eq:regcond1}
\end{align}
\begin{proof}
    There we only give the proof of  \eqref{eq:regcond1} for $i=1$, as it follows a similar route when $i=2$. From \eqref{eq:nmcond2}
    \begin{align*}
    \sigma_1(\vZ_{\star{1}}^k)
    &\leq\big(1+6.4\kappa^{-1}\varepsilon_0\sum_{t=0}^\infty(1-\frac{11\eta^{\prime}}{100\kappa})^t\big)\sigma_1\(\vM_\star\)^{\frac{1}{2}}\\
    &=(1+\frac{640\varepsilon_0}{11\eta^{\prime}})\sigma_1(\vM_\star)^{\frac{1}{2}}\leq\sqrt{2}\sigma_1(\vM_\star)^{\frac{1}{2}},    
    \end{align*}
    where the last inequality is from that $\varepsilon_0$ is small enough such that $\frac{640\varepsilon_0}{11\eta^{\prime}}\leq\sqrt{2}-1$ and $\eta^{\prime}$ is some constant.
    From \eqref{eq:nmcond1}
    \begin{align*}
      \sigma_r(\vZ_{\star{1}}^k)&\geq\big(1-6.4\kappa^{-1}\varepsilon_0\sum_{t=0}^\infty(1-\frac{11\eta^{\prime}}{100\kappa})^t\big)\sigma_r\(\vM_\star\)^{\frac{1}{2}}\\
      &=(1-\frac{640\varepsilon_0}{11\eta^{\prime}})\geq\frac{\sqrt{2}}{2}\sigma_r\(\vM_\star\)^{\frac{1}{2}},
    \end{align*}
    as long as $\varepsilon_0$ is a small constant such that $\frac{640\varepsilon_0}{11\eta^{\prime}}\leq1-\frac{\sqrt{2}}{2}$.
\end{proof}
\end{claim}
Now we begin to establish \eqref{eq:dist_contraction_npj}.
According to the definition of the distance metric \eqref{eq:def_dist_invert}
\begin{small}
 \begin{align*}
  \distPsq{\tilde{\vZ}^{k+1}}{\vZ_\star}&\leq\|\tilde{\vZ}^{k+1}-\vZ_{\star{1}}^k\|_F^2+\|\tilde{\vZ}^{k+1}-\vZ_{\star{2}}^k\|_F^2\\
  &=\distPsq{\vZ^k}{\vZ_\star}-2\eta\Real\langle\nabla f^{(k)}(\vZ^k),\vDeltaa^k+\vDeltab^k\rangle\\
  &\quad+2\eta^2\|\nabla f^{(k)}(\vZ^k)\|_F^2, \numberthis \label{eq:dist_iter_oracle} \end{align*}
\end{small}\noindent{where} the equality results from $\tilde{\vZ}^{k+1}=\vZ^k-\eta\nabla f^{(k)}(\vZ^k)$ and we set $\vDeltaa^k=\vZ^k-\vZ_{\star{1}}^k$, $\vDeltab^k=\vZ^k-\vZ_{\star{2}}^k$.
 Next we build a lower bound for $\Real\langle\nabla f^{(k)}(\vZ^k),\vDeltaa^k+\vDeltab^k\rangle$ and an upper bound for $\|\nabla f^{(k)}(\vZ^k)\|_F^2$ in step 1 and step 2. Finally, we give the upper bound of \eqref{eq:dist_iter_oracle} in step 3. 
 

\textbf{Step 1: bounding $\Real\la\nabla f^{(k)}(\vZ^k),\vDeltaa^k+\vDeltab^k\ra$.} We reformulate this term into two parts to bound.
\begin{small}
\begin{align*}
  &\Real\la\nabla f^{(k)}(\vZ^k),\vDeltaa^k+\vDeltab^k\ra\\
  &=
  \Real \langle(\I-\G\G^*)(\vZ^k(\vZ^k)^T-\vM_\star)\bar{\vZ}^k \\
  &\quad+\hat{p}^{-1}\G\P_{{\hat{\Omega}}_k}\G^*(\vZ^k(\vZ^k)^T-\vM_\star)\bar{\vZ}^k,\vDeltaa^k+\vDeltab^k \rangle\\
  &=\Real\langle\vZ^k(\vZ^k)^T-\vM_\star ,\vDeltaa^k(\vZ^k)^T+\vZ^k(\vDeltab^k)^T\rangle+\Real\\&\quad\langle\G(\hat{p}^{-1}\P_{{\hat{\Omega}}_k}-\I)\G^*(\vZ^k(\vZ^k)^T-\vM_\star),\vDeltaa^k(\vZ^k)^T+\vZ^k(\vDeltab^k)^T\rangle\\
  &=I_1+I_2,
\end{align*}
\end{small}where the first equality results from $(\I-\G\G^*)(\vZ^k(\vZ^k)^T-\vM_\star)=(\I-\G\G^*)(\vZ^k(\vZ^k)^T)$, and the second equality results from the symmetry of $\vZ^k(\vZ^k)^T-\vM_\star$. In the following derivations, we invoke the decomposition as
\begin{small}
\begin{align*}
\vDeltaa^k(\vZ^k)^T+\vZ^k(\vDeltab^k)^T&=\vDeltaa^k(\vZ_{\star{2}}^k)^T+\vZ_{\star{1}}^k(\vDeltab^k)^T+2\vDeltaa^k(\vDeltab^k)^T,\\
\vZ^k(\vZ^k)^T-\vM_\star
&=\vDeltaa^k(\vZ_{\star{2}}^k)^T+\vZ_{\star{1}}^k(\vDeltab^k)^T+\vDeltaa^k(\vDeltab^k)^T.
\end{align*}
\end{small}

Thus for $I_1$, we have
\begin{small}
\begin{align*}
I_1&=\Real\langle\vDeltaa^k(\vZ_{\star{2}}^k)^T+\vZ_{\star{1}}^k(\vDeltab^k)^T+\vDeltaa^k(\vDeltab^k)^T,\\
&\quad\vDeltaa^k(\vZ_{\star{2}}^k)^T+\vZ_{\star{1}}^k(\vDeltab^k)^T+2\vDeltaa^k(\vDeltab^k)^T\rangle\\
    &=\|\vDeltaa^k(\vZ_{\star{2}}^k)^T+\vZ_{\star{1}}^k(\vDeltab^k)^T\|_F^2+2\|\vDeltaa^k(\vDeltab^k)^T\|_F^2\\
    &\quad+3\Real\langle\vDeltaa^k(\vZ_{\star{2}}^k)^T+\vZ_{\star{1}}^k(\vDeltab^k)^T,\vDeltaa^k(\vDeltab^k)^T\rangle\\
    &\geq\frac{3}{4}\|\vDeltaa^k(\vZ_{\star{2}}^k)^T{+}\vZ_{\star{1}}^k(\vDeltab^k)^T\|_F^2-7\|\vDeltaa^k\|_F^2\|\vDeltab^k\|_F^2\\
    &\geq \frac{3}{4}\|\vDeltaa^k(\vZ_{\star{2}}^k)^T{+}\vZ_{\star{1}}^k(\vDeltab^k)^T\|_F^2
    \\&\quad-14\varepsilon_0^2\sigma_r(\vM_\star)(\|\vDeltaa^k\|_F^2+\|\vDeltab^k\|_F^2),
\end{align*}
\end{small}where the third inequality results from $3\Real\langle\vA,\vB\rangle=\|\frac{1}{2}\vA+3\vB\|_F^2-\frac{1}{4}\|\vA\|_F^2-9\|\vB\|_F^2\geq-\frac{1}{4}\|\vA\|_F^2-9\|\vB\|_F^2$. 
The last inequality follows from
\begin{align*}
    \mathrm{max}\{\|\vDeltaa^k\|_F,\|\vDeltab^k\|_F\}{\leq}\distP{\vZ^k}{\vZ_\star}{\leq}2\varepsilon_0\sqrt{\sigma_r(\vM_\star)}. \numberthis \label{eq:converge_intermid_deltaab}
\end{align*}

For $I_2$, we have
\begin{small}
\begin{align*}
I_2&=\Real\langle\G(\hat{p}^{-1}\P_{{\hat{\Omega}}_k}-\I)\G^*(\vDeltaa^k(\vZ_{\star{2}}^k)^T+\vZ_{\star{1}}^k(\vDeltab^k)^T+\vDeltaa^k(\vDeltab^k)^T),\\
&\quad\vDeltaa^k(\vZ_{\star{2}}^k)^T+\vZ_{\star{1}}^k(\vDeltab^k)^T+2\vDeltaa^k(\vDeltab^k)^T\rangle\\
&=\Real\langle\G(\hat{p}^{-1}\P_{{\hat{\Omega}}_k}-\I)\G^*(\vDeltaa^k(\vZ_{\star{2}}^k)^T+\vZ_{\star{1}}^k(\vDeltab^k)^T),\\
&\quad\vDeltaa^k(\vZ_{\star{2}}^k)^T+\vZ_{\star{1}}^k(\vDeltab^k)^T\rangle\\&
+3\Real\langle\G(\hat{p}^{-1}\P_{{\hat{\Omega}}_k}-\I)\G^*(\vDeltaa^k(\vDeltab^k)^T),
\vDeltaa^k(\vZ_{\star{2}}^k)^T+\vZ_{\star{1}}^k(\vDeltab^k)^T\rangle\\
&+2\Real\langle\G(\hat{p}^{-1}\P_{{\hat{\Omega}}_k}-\I)\G^*(\vDeltaa^k(\vDeltab^k)^T),\vDeltaa^k(\vDeltab^k)^T\rangle\\
&=I_{2,1}+3I_{2,2}+2I_{2,3}.
\end{align*}
\end{small}
It is easily to see  $\vDeltaa^k(\vZ_{\star{2}}^k)^T+\vZ_{\star{1}}^k(\vDeltab^k)^T\in T$. Invoking Lemma~\ref{lem:tangent_completion_inq} when $\hat{m}\gtrsim O\( \varepsilon_0^{-2}\mu r\log(n)\)$, we have
\begin{small}
\begin{align*}
    I_{2,1}&\geq-\|\P_T\G(\hat{p}^{-1}\P_{{\hat{\Omega}}_k}-\I)\G^*\P_T\|\|\vDeltaa^k(\vZ_{\star{2}}^k)^T+\vZ_{\star{1}}^k(\vDeltab^k)^T\|_F^2\\
    &\geq-\varepsilon_0\|\vDeltaa^k(\vZ_{\star{2}}^k)^T+\vZ_{\star{1}}^k(\vDeltab^k)^T\|_F^2.
\end{align*}
\end{small}
We have the following results from Claim \ref{claim:regcond_for_factors} 
\begin{align*}
\|\vZ_{\star{1}}^k\|_{2,\infty}\leq\|\vU_\star\|_{2,\infty}\|\vZ_{\star{1}}^k\|\leq 2\sqrt{\frac{\mu r\sigma}{n}},\\
\|\vDeltaa^k\|_{2,\infty}\leq\|\vZ^k\|_{2,\infty}+\|\vZ_{\star{1}}^k\|_{2,\infty}\leq 4\sqrt{\frac{\mu r\sigma}{n}},
\end{align*}
and similarly, the bounds hold for $\|\vZ_{\star{2}}^k\|_{2,\infty}$ and $\|\vDeltab^k\|_{2,\infty}$.

combining the previous results and invoking Lemma~\ref{lem:Hankcom_tr_inq}
\begin{align*}
    I_{2,2}&=\Real\langle\G(\hat{p}^{-1}\P_{{\hat{\Omega}}_k}-\I)\G^*(\vDeltaa^k(\vDeltab^k)^T),
\vDeltaa^k(\vZ_{\star{2}}^k)^T\rangle\\
&\quad+\Real\langle\G(\hat{p}^{-1}\P_{{\hat{\Omega}}_k}-\I)\G^*(\vDeltaa^k(\vDeltab^k)^T),
\vZ_{\star{1}}^k(\vDeltab^k)^T\rangle\\
&\geq -c_3\sqrt{\frac{n^2\log(n)}{\hat{m}}}(\|\vDeltaa^k\|_{2,\infty}\|\vDeltab^k\|_F\|\vDeltaa^k\|_F\|\vZ_{\star{2}}^k\|_{2,\infty}\\
&\quad+\|\vDeltaa^k\|_F\|\vDeltab^k\|_{2,\infty}\|\vZ_{\star{1}}^k\|_{2,\infty}\|\vDeltab^k\|_F)\\
&\overset{(a)}{\geq} -16c_3\sqrt{\frac{\mu^2 r^2(1+\varepsilon_0)^2\log(n)\sigma_1^2(\vM_\star)}{\hat{m}(1-\varepsilon_0)^2}}\|\vDeltaa^k\|_F\|\vDeltab^k\|_F\\
&\overset{(b)}{\geq} -\varepsilon_0\sigma_r(\vM_\star)\|\vDeltaa^k\|_F\|\vDeltab^k\|_F,
\end{align*}
$(a)$ results from $\sigma=\frac{\sigma_1(\vM^0)}{1-\varepsilon_0}\leq\frac{1+\varepsilon_0}{1-\varepsilon_0}\sigma_1(\vM_\star)$, where $\varepsilon_0$ is a small constant. $(b)$ follows from  $\hat{m}\gtrsim O(\varepsilon_0^{-2}\mu^2\kappa^2 r^2\log(n))$. 

Similarly invoking Lemma~\ref{lem:Hankcom_tr_inq}, one can obtain 
\begin{align*}
I_{2,3}\geq -\varepsilon_0\sigma_r(\vM_\star)\|\vDeltaa^k\|_F\|\vDeltab^k\|_F.  
\end{align*}
Therefore
\begin{small}
\begin{align*}
    I_2&{\geq} {-}\varepsilon_0\|\vDeltaa^k(\vZ_{\star{2}}^k)^T{+}\vZ_{\star{1}}^k(\vDeltab^k)^T\|_F^2{-}5\varepsilon_0\sigma_r(\vM_\star)\|\vDeltaa^k\|_F\|\vDeltab^k\|_F\\
    &{\geq} {-}\varepsilon_0\|\vDeltaa^k(\vZ_{\star{2}}^k)^T{+}\vZ_{\star{1}}^k(\vDeltab^k)^T\|_F^2
    {-}\frac{5}{2}\varepsilon_0\sigma_r(\vM_\star)(\|\vDeltaa^k\|_F^2+\|\vDeltab^k\|_F^2),
\end{align*}
\end{small}and combine $I_1$, $I_2$ to obtain
    \begin{small}    
\begin{align*}
     \Real\langle\nabla f^{(k)}(\vZ),\vDeltaa^k+&\vDeltab^k\rangle\geq (\frac{3}{4}-\varepsilon_0)\|\vDeltaa^k(\vZ_{\star{2}}^k)^T+\vZ_{\star{1}}^k(\vDeltab^k)^T\|_F^2\\
    &-(14\varepsilon_0^2+\frac{5}{2}\varepsilon_0)\sigma_r(\vM_\star)(\|\vDeltaa^k\|_F^2+\|\vDeltab^k\|_F^2).
\end{align*}
    \end{small}We point out a result towards the first term in the RHS of the previous inequality. 
\begin{small}
\begin{align*}
&\|\vDeltaa^k(\vZ_{\star{2}}^k)^T+\vZ_{\star{1}}^k(\vDeltab^k)^T\|_F^2\\&=\|\vDeltaa^k(\vZ_{\star{2}}^k)^T\|_F^2+\|\vZ_{\star{1}}^k(\vDeltab^k)^T\|_F^2+2\Real\langle\vDeltaa^k(\vZ_{\star{2}}^k)^T,\vZ_{\star{1}}^k(\vDeltab^k)^T\rangle\\
&\overset{(a)}{\geq}\|\vDeltaa^k(\vZ_{\star{2}}^k)^T\|_F^2+\|\vZ_{\star{1}}^k(\vDeltab^k)^T\|_F^2\\
&\geq\sigma_r^2(\vZ_{\star{1}}^k)\|\vDeltab^k\|_F^2+\sigma_r^2(\vZ_{\star{2}}^k)\|\vDeltaa^k\|_F^2\\&\overset{(b)}{\geq}\frac{1}{2}\sigma_r(\vM_\star)(\|\vDeltaa^k\|_F^2+\|\vDeltab^k\|_F^2),\numberthis \label{eq:distcontrac_opt_funda_term}
\end{align*}
\end{small}where $(a)$ results from a corollary of the first-order optimality condition \eqref{eq:optcond_primal}: $(\vZ_{\star{1}}^k)^{H}\vDeltaa^k=(\vDeltab^k)^T{\overline{\vZ_{\star{2}}^k}}$, which is
\begin{small}
\begin{align*}
\langle\vDeltaa^k(\vZ_{\star{2}}^k)^T,\vZ_{\star{1}}^k(\vDeltab^k)^T\rangle=\langle(\vZ_{\star{1}}^k)^{H}\vDeltaa^k,(\vDeltab^k)^T{\overline{\vZ_{\star{2}}^k}}\rangle\geq 0,
\end{align*}
\end{small}and $(b)$ comes from Claim \ref{claim:regcond_for_factors}.

\textbf{Step 2: bounding $\|\nabla f^{(k)}(\vZ^k)\|_F^2$.} We bound this term by controlling two parts $I_3$ and $I_4$.
\begin{small}
    \begin{align*}
    &\|\nabla f^{(k)}(\vZ^k)\|_F^2
    \\&{=}\|\big((\vZ^k(\vZ^k)^T{-}\vM_\star){+}\G(\hat{p}^{-1}\P_{{\hat{\Omega}}_k}-\I)\G^*(\vZ^k(\vZ^k)^T{-}\vM_\star)\big)\bar{\vZ}^k\|_F^2\\
    &\leq2\|(\vZ^k(\vZ^k)^T-\vM_\star)\bar{\vZ}^k\|_F^2\\
    &\quad+2\|\G(\hat{p}^{-1}\P_{{\hat{\Omega}}_k}-\I)\G^*(\vZ^k(\vZ^k)^T-\vM_\star)\bar{\vZ}^k\|_F^2=2(I_3^2+I_4^2),
\end{align*}
\end{small}where the first equality results from the fact $(\I-\G\G^*)(\vZ^k(\vZ^k)^T-\vM_\star)=(\I-\G\G^*)(\vZ^k(\vZ^k)^T)$.
\begin{small}
\begin{align*}
    I_3^2&\leq \|\vZ^k\|^2\|\vZ^k(\vZ^k)^T-\vM_\star\|_F^2\\
    &\leq 2\|\vZ^k\|^2\(\|\vDeltaa^k(\vZ_{\star{2}}^k)^T+\vZ_{\star{1}}^k(\vDeltab^k)^T\|_F^2+\|\vDeltaa^k\|_F^2\|\vDeltab^k\|_F^2\),\\
    I_4&=\|\G(\hat{p}^{-1}\P_{{\hat{\Omega}}_k}-\I)\G^*(\vZ^k(\vZ^k)^T-\vM_\star)\bar{\vZ}^k\|_F\\
    &=\max_{\|\tilde{\vX}\|_F=1}|\langle\G(\hat{p}^{-1}\P_{{\hat{\Omega}}_k}-\I)\G^*(\vZ^k(\vZ^k)^T-\vM_\star),\tilde{\vX}(\vZ^k)^T\rangle|.
\end{align*}
\end{small}
Note that 
\begin{small}
\begin{align*}
    &|\langle\G(\hat{p}^{-1}\P_{{\hat{\Omega}}_k}-\I)\G^*(\vZ^k(\vZ^k)^T-\vM_\star),\tilde{\vX}(\vZ^k)^T\rangle|\\
    &=|\langle\G(\hat{p}^{-1}\P_{{\hat{\Omega}}_k}-\I)\G^*(\vDeltaa^k(\vZ_{\star{2}}^k)^T+\vZ_{\star{1}}^k(\vDeltab^k)^T+\vDeltaa^k(\vDeltab^k)^T),\\
    &\quad\tilde{\vX}(\vDeltab^k)^T+\tilde{\vX}(\vZ_{\star{2}}^k)^T\rangle|\\
    &\leq |\langle\G(\hat{p}^{-1}\P_{{\hat{\Omega}}_k}-\I)\G^*(\vDeltaa^k(\vZ_{\star{2}}^k)^T+\vZ_{\star{1}}^k(\vDeltab^k)^T),\tilde{\vX}(\vZ_{\star{2}}^k)^T\rangle|\\
    &\quad+|\langle\G(\hat{p}^{-1}\P_{{\hat{\Omega}}_k}-\I)\G^*(\vDeltaa^k(\vZ_{\star{2}}^k)^T+\vZ_{\star{1}}^k(\vDeltab^k)^T),\tilde{\vX}(\vDeltab^k)^T\rangle|\\
    &\quad+|\langle\G(\hat{p}^{-1}\P_{{\hat{\Omega}}_k}-\I)\G^*(\vDeltaa^k(\vDeltab^k)^T),\tilde{\vX}(\vDeltab^k)^T\rangle|\\
    &\quad+ |\langle\G(\hat{p}^{-1}\P_{{\hat{\Omega}}_k}-\I)\G^*(\vDeltaa^k(\vDeltab^k)^T),\tilde{\vX}(\vZ_{\star{2}}^k)^T\rangle|
    \\
    &= I_{4,1}+I_{4,2}+I_{4,3}+I_{4,4}.
\end{align*}
\end{small}

As $\vDeltaa^k(\vZ_{\star{2}}^k)^T+\vZ_{\star{1}}^k(\vDeltab^k)^T\in T$ and $\tilde{\vX}(\vZ_{\star{2}}^k)^T\in T$, we invoke Lemma~\ref{lem:tangent_completion_inq} to obtain the upper bound of $I_{4,1}$ when $\|\tilde{\vX}\|_F=1$, similarly as bounding $I_{2,1}$
\begin{small}
\begin{align*}
    I_{4,1}&\leq \|\P_T\G(\hat{p}^{-1}\P_{{\hat{\Omega}}_{k}}-\I)\G^*\P_T\|\\
    &\quad\cdot\|\vDeltaa^k(\vZ_{\star{2}}^k)^T+\vZ_{\star{1}}^k(\vDeltab^k)^T\|_F\cdot\|\tilde{\vX}(\vZ_{\star{2}}^k)^T\|_F\\
    &\leq \varepsilon_0 \sqrt{2\sigma_1(\vM_\star)}\|\vDeltaa^k(\vZ_{\star{2}}^k)^T+\vZ_{\star{1}}^k(\vDeltab^k)^T\|_F.
\end{align*}
\end{small}Similarly as bounding $I_{2,2}$, when $\hat{m}\gtrsim O(\varepsilon_0^{-2}\mu^2\kappa^2 r^2\log(n))$ and $\|\tilde{\vX}\|_F=1$, invoke Lemma~\ref{lem:Hankcom_tr_inq} to obtain 
\begin{align*}
    I_{4,2}
   &\leq c_3\sqrt{\frac{n^2\log(n)}{\hat{m}}}(\|\vDeltaa^k\|_{F}\|\vZ_{\star{2}}^k\|_{2,\infty}\|\tilde{\vX}\|_F\|\vDeltab^k\|_{2,\infty}\\
&\quad+\|\vZ_{\star{1}}^k\|_{2,\infty}\|\vDeltab^k\|_{F}\|\tilde{\vX}\|_F\|\vDeltab^k\|_{2,\infty})\\
&\leq \varepsilon_0\sigma_r(\vM_\star)
(\|\vDeltaa^k\|_F+\|\vDeltab^k\|_F),
\\
    I_{4,3}&\leq c_3\sqrt{\frac{n^2\log(n)}{\hat{m}}}\|\vDeltaa^k\|_{2,\infty}\|\vDeltab^k\|_F\|\vDeltab^k\|_{2,\infty}\|\tilde{\vX}\|_F\\
    &\leq \varepsilon_0\sigma_r(\vM_\star)
    \|\vDeltab^k\|_F,
    \\I_{4,4}&\leq c_3\sqrt{\frac{n^2\log(n)}{\hat{m}}}\|\vDeltaa^k\|_F\|\vDeltab^k\|_{2,\infty}\|\vZ_{\star{2}}^k\|_{2,\infty}\|\tilde{\vX}\|_F\\
    &\leq \varepsilon_0\sigma_r(\vM_\star)
    \|\vDeltaa^k\|_F.
\end{align*}

Combine $I_{4,1}$, $I_{4,2}$, $I_{4,3}$ and  $I_{4,4}$ to obtain 
\begin{small}
\begin{align*}
  I_4^2&\leq 4\varepsilon_0^2\sigma_1(\vM_\star)\Big(\|\vDeltaa^k(\vZ_{\star{2}}^k)^T+\vZ_{\star{1}}^k(\vDeltab^k)^T\|_F^2\\&\quad+4\sigma_r(\vM_\star)(\|\vDeltaa^k\|_F^2+\|\vDeltab^k\|_F^2)\Big).
\end{align*}
\end{small}
Consequently, combine $I_3^2$ and $I_4^2$ to obtain
\begin{small}
\begin{align*}
  &\|\nabla f^{(k)}(\vZ^k)\|_F^2\\
  &\overset{(a)}{\leq }\sigma_1(\vM_\star)\Big((10+8\varepsilon_0^2)\|\vDeltaa^k(\vZ_{\star{2}}^k)^T+\vZ_{\star{1}}^k(\vDeltab^k)^T\|_F^2\\
&\quad+10\|\vDeltaa^k\|_F^2\|\vDeltab^k\|_F^2+32\varepsilon_0^2\sigma_r(\vM_\star)(\|\vDeltaa^k\|_F^2+\|\vDeltab^k\|_F^2)\Big)\\
&\overset{(b)}{\leq }\sigma_1(\vM_\star)\Big((10+8\varepsilon_0^2)\|\vDeltaa^k(\vZ_{\star{2}}^k)^T+\vZ_{\star{1}}^k(\vDeltab^k)^T\|_F^2\\
&\quad+52\varepsilon_0^2\sigma_r(\vM_\star)(\|\vDeltaa^k\|_F^2+\|\vDeltab^k\|_F^2)\Big),
\end{align*}
\end{small}
where step $(a)$ comes from 
the fact that
\begin{align*}
 \|\vZ^k\|^2&\leq\(\|\vDeltaa^k\|+\|\vZ_{\star{1}}^k\|\)^2
 \leq2.5\sigma_1(\vM_\star),   
\end{align*}
 where we invoke the results \eqref{eq:converge_intermid_deltaab} and Claim \ref{claim:regcond_for_factors} for $\varepsilon_0$ is a small constant. Step $(b)$ follows from \eqref{eq:converge_intermid_deltaab} directly.

Finally, we give the upper bound of \eqref{eq:dist_iter_oracle} by combining step 1 and step 2.

\textbf{Step 3:
bounding \eqref{eq:dist_iter_oracle}.}
Combining the above bounds for $\Real\la\nabla f^{(k)}(\vZ^k),\vDeltaa^k+\vDeltab^k\ra$ and $\|\nabla f^{(k)}(\vZ^k)\|_F^2$, we have
\begin{small}
\begin{align*}
    &-2\eta\Real\la\nabla f^{(k)}(\vZ^k),\vDeltaa^k+\vDeltab^k\ra+2\eta^2\|\nabla f^{(k)}(\vZ^k)\|_F^2\\
    &\leq\frac{\alpha}{\sigma_1(\vM_\star)}\|\vDeltaa^k(\vZ_{\star{2}}^{k})^T+\vZ_{\star{1}}^k(\vDeltab^k)^T\|_F^2+\frac{\beta}{\kappa}(\|\vDeltaa^k\|_F^2+\|\vDeltab^k\|_F^2)\\
    &\overset{(a)}{\leq}\frac{\alpha}{2\kappa}(\|\vDeltaa^k\|_F^2+\|\vDeltab^k\|_F^2)+\frac{\beta}{\kappa}(\|\vDeltaa^k\|_F^2+\|\vDeltab^k\|_F^2)\\
    &\overset{(b)}{\leq}-\frac{11\eta^\prime}{50\kappa}\distPsq{\vZ^k}{\vZ_\star},
\end{align*}
\end{small}where $\eta=\frac{\eta^{\prime}}{\sigma_1(\vM_\star)}$, $\alpha=4{\eta^{\prime}}^2(5+4\varepsilon_0^2)-\eta^{\prime}(1.5-2\varepsilon_0)$ and $\beta=\eta^{\prime}(28\varepsilon_0^2+5\varepsilon_0)+104{\eta^{\prime}}^2\varepsilon_0^2$. For $\eta^{\prime}\leq\frac{1}{54}$, we list the following numerical results for $\varepsilon_0$ is a small enough constant 
\begin{align*}
\alpha&=4{\eta^{\prime}}^2(5+4\varepsilon_0^2)-\eta^{\prime}(1.5-2\varepsilon_0)\leq 0,\\
\frac{\alpha}{2}+\beta&=2{\eta^{\prime}}^2(5+56\varepsilon_0^2)+\eta^{\prime}(28\varepsilon_0^2+6\varepsilon_0-0.75)\leq -\frac{11\eta^{\prime}}{50}.
\end{align*}
Therefore step $(a)$ comes from \eqref{eq:distcontrac_opt_funda_term} and   $\alpha\leq 0$, step $(b)$ follows from 
$\frac{\alpha}{2}+\beta\leq -\frac{11\eta^{\prime}}{50}$.

Finally, we obtain an upper bound of  \eqref{eq:dist_iter_oracle}
\begin{align*}
  \distPsq{\tilde{\vZ}^{k+1}}{\vZ_\star}&\leq\|\tilde{\vZ}^{k+1}-\vZ_{\star{1}}^k\|_F^2+\|\tilde{\vZ}^{k+1}-\vZ_{\star{2}}^k\|_F^2\\
  &\leq(1-\frac{11\eta^\prime}{50\kappa})\distPsq{\vZ^k}{\vZ_\star},\numberthis \label{eq:dist_contrac_full}
\end{align*}
and further establish \eqref{eq:dist_contraction_npj}
\begin{align*}
 \distP{\tilde{\vZ}^{k+1}}{\vZ_\star}
 \leq(1-\frac{11\eta^{\prime}}{100\kappa})\distP{\vZ^k}{\vZ_\star}. 
\end{align*}
 
\section{} 

\subsection{Proof of the induction hypotheses from the $ k$-th to the $(k+1)$-th step}\label{apd:pf_induction}
As \eqref{eq:distancebd}, \eqref{eq:nmcond1}, and \eqref{eq:nmcond2} hold for the $k$-th step, we have the following linear convergence result by invoking Lemma~\ref{lem:dist_contract} 
    \begin{align*}
 \distP{\tilde{\vZ}^{k+1}}{\vZ_\star} \leq (1-\frac{11\eta^{\prime}}{100\kappa})\distP{\vZ^k}{\vZ_\star},\numberthis\label{eq:induc_linconverg}
\end{align*}
with probability at least $1-c_2 n^{-2}$ when $\hat{m}\gtrsim O\( \varepsilon_0^{-2}\mu^2\kappa^2 r^2\log(n)\)$, $\eta=\frac{\eta^{\prime}}{\sigma_1(\vM_\star)}$, $\eta^{\prime}$ and $\varepsilon_0$ are some appropriate constants.
We establish \eqref{eq:distancebd} for the $(k+1)$-th step by proving the following non-expansiveness property of the projection in terms of the distance metric:
\begin{align*}
    \distP{\vZ^{k+1}}{\vZ_\star}\leq\distP{\tilde{\vZ}^{k+1}}{\vZ_\star},
\end{align*}
where $\vZ^{k+1}=\PC{\tilde{\vZ}^{k+1}}$. 

Set $\tilde{\vZ}_{\star{1}}^{k+1}=\vZ_\star\tilde{\vP}^{k+1}$, $\tilde{\vZ}_{\star{2}}^{k+1}=\vZ_\star(\tilde{\vP}^{k+1})^{-T}$ where $\tilde{\vP}^{k+1}$ is the optimal solution as defined in \eqref{eq:def_optinvP}. 
According to \eqref{eq:induc_linconverg} and the definition of $\distP{\tilde{\vZ}^{k+1}}{\vZ_\star}$
\begin{small}
 \begin{align*}
     \|\tilde{\vZ}^{k+1}-\tilde{\vZ}_{\star{1}}^{k+1}\|\leq\distP{\tilde{\vZ}^{k+1}}{\vZ_\star}
     \leq(1-\frac{11\eta^{\prime}}{100\kappa})\distP{\vZ^k}{\vZ_\star}.\numberthis\label{eq:nmcond_triangle1}
     \end{align*}   
\end{small}

     One simple corollary of \eqref{eq:dist_contrac_full} is that
     \begin{align*}
     \|\tilde{\vZ}^{k+1}-\vZ_{\star{1}}^k\|&\leq(1-\frac{11\eta^{\prime}}{100\kappa}) \distP{\vZ^k}{\vZ_\star}. \numberthis\label{eq:nmcond_triangle2}
\end{align*}
Invoking triangle inequality and combining \eqref{eq:nmcond_triangle1},\eqref{eq:nmcond_triangle2}, and \eqref{eq:nmcond2}, we obtain that 
\begin{align*}
    \sigma_1(\tilde{\vZ}_{\star{1}}^{k+1})&\leq \|\vZ_{\star{1}}^k\|+\|\tilde{\vZ}_{\star{1}}^{k+1}-\tilde{\vZ}^{k+1}\|+\|\tilde{\vZ}^{k+1}-\vZ_{\star{1}}^k\|\\
    &\leq\Big[1+6.4\kappa^{-1}\varepsilon_0\sum_{t=0}^{\infty}(1-\frac{11\eta^\prime}{100\kappa})^t\Big]\sigma_1\(\vM_\star\)^{\frac{1}{2}}\\
    &=(1+\frac{640\varepsilon_0}{11\eta^{\prime}})\sigma_1\(\vM_\star\)^{\frac{1}{2}}\leq\sqrt{2}\sigma_1\(\vM_\star\)^{\frac{1}{2}},
\end{align*}
when $\varepsilon_0$ is a small constant such that $\frac{640\varepsilon_0}{11\eta^{\prime}}\leq\sqrt{2}-1$. 

As $\tilde{\vZ}_{\star{1}}^{k+1}=\vZ_\star\tilde{\vP}^{k+1}$, one can obtain
\begin{align*}
\|\tilde{\vZ}_{\star{1}}^{k+1}\|_{2,\infty}\leq \|\vU_\star\|_{2,\infty}\cdot\sigma_1(\tilde{\vZ}_{\star{1}}^{k+1})\leq 2\sqrt{\frac{\mu r\sigma}{n}}.
\end{align*}
Thus $\tilde{\vZ}_{\star{1}}^{k+1}\in\CS$ and one can also obtain $\tilde{\vZ}_{\star{2}}^{k+1}\in\CS$ , following a similar route.

combining the definition of the distance metric and  the non-expansiveness of convex projection $\P_\CS$, we have
\begin{align*}
    \distP{\vZ^{k+1}}{\vZ_\star}&{\leq} \sqrt{\|\vZ^{k+1}-\tilde{\vZ}_{\star{1}}^{k+1}\|_F^2+\|\vZ^{k+1}-\tilde{\vZ}_{\star{2}}^{k+1}\|_F^2}\\
    &{\leq} \sqrt{\|\tilde{\vZ}^{k+1}-\tilde{\vZ}_{\star{1}}^{k+1}\|_F^2+\|\tilde{\vZ}^{k+1}-\tilde{\vZ}_{\star{2}}^{k+1}\|_F^2}\\
    &=\distP{\tilde{\vZ}^{k+1}}{\vZ_\star}.\numberthis\label{eq:dist_contrac_k+1}
\end{align*}

The things left are to prove \eqref{eq:nmcond1} and \eqref{eq:nmcond2} for the $(k+1)$-th step. We only give the proof of bounding $\vZ_{\star{1}}^{k+1}$ in details, because bounding  $\vZ_{\star{2}}^{k+1}$ follows a similar route.

From \eqref{eq:induc_linconverg} and \eqref{eq:dist_contrac_k+1}, one can obtain that 
\begin{small}
\begin{align*}
   \|\vZ_{\star{1}}^{k+1}-\vZ^{k+1}\|&\leq\distP{\vZ^{k+1}}{\vZ_\star}\leq(1-\frac{11\eta^{\prime}}{100\kappa})\distP{\vZ^k}{\vZ_\star},\numberthis\label{eq:nmcond_triangle3}\\
   \|\vZ^{k+1}-\tilde{\vZ}_{\star{1}}^{k+1}\|&\leq(1-\frac{11\eta^{\prime}}{100\kappa})\distP{\vZ^k}{\vZ_\star}.\numberthis\label{eq:nmcond_triangle4}
\end{align*}
\end{small}

Invoking the triangle inequality for several times, and combining the results \eqref{eq:nmcond_triangle1}, 
 \eqref{eq:nmcond_triangle2}, \eqref{eq:nmcond_triangle3}, and \eqref{eq:nmcond_triangle4}, we have
\begin{align*}
 \sigma_1(\vZ_{\star{1}}^{k+1})&=\|\vZ_{\star{1}}^{k+1}\|\\
 &\leq\|\vZ_{\star{1}}^k\|+\|\vZ_{\star{1}}^{k+1}-\vZ_{\star{1}}^{k}\|
 \\ &\leq\sigma_1(\vZ_{\star{1}}^k)+\|\vZ_{\star{1}}^{k+1}-\vZ^{k+1}\|
 +\|\vZ^{k+1}-\tilde{\vZ}_{\star{1}}^{k+1}\|\\&\quad+\|\tilde{\vZ}_{\star{1}}^{k+1}-\tilde{\vZ}^{k+1}\|+\|\tilde{\vZ}^{k+1}-\vZ_{\star{1}}^k\|\\
    &\leq\Big[1+6.4\kappa^{-1}\varepsilon_0\sum_{t=0}^{k+1}(1-\frac{11\eta^{\prime}}{100\kappa})^t\Big]\sigma_1\(\vM_\star\)^{\frac{1}{2}}.
    \end{align*}
    Besides, invoke Weyl's Theorem to obtain that  
\begin{align*}
&\sigma_r\(\vZ_{\star{1}}^{k+1}\)\\
    &\geq\sigma_r\(\vZ_{\star{1}}^{k}\)-\|\vZ_{\star{1}}^{k+1}-\vZ_{\star{1}}^{k}\|\\
    &\geq\sigma_r\(\vZ_{\star{1}}^{k}\)-\Big(\|\vZ_{\star{1}}^{k+1}-\vZ^{k+1}\|+\|\vZ^{k+1}-\tilde{\vZ}_{\star{1}}^{k+1}\|\\
    &\quad+\|\tilde{\vZ}_{\star{1}}^{k+1}-\tilde{\vZ}^{k+1}\|+\|\tilde{\vZ}^{k+1}-\vZ_{\star{1}}^k\|\Big)\\
    &\geq\Big[1-6.4\kappa^{-1}\varepsilon_0\sum_{t=0}^{k+1}(1-\frac{11\eta^{\prime}}{100\kappa})^t\Big]\sigma_r\(\vM_\star\)^{\frac{1}{2}}. 
\end{align*}
Therefore, \eqref{eq:nmcond1} and \eqref{eq:nmcond2} for the $(k+1)$-step are established. 

\subsection{Proof of Remark~\ref{rmk:recoveryerr_vec}}\label{apd:pf_rmkrecov_vec} 

We establish a relationship between the recovery error and our distance metric, which is 
  \begin{align*}
        \|\vx^k-\vx\|_2&\leq \|\vy^k-\vy\|_2=\|\G^*(\vZ^k(\vZ^k)^T-\G\vy)\|_2\\
        &\leq\|\vZ^k{(\vZ^k)}^T-\vM_\star\|_F\\&=\|(\vZ^k-\vZ^k_{\star{1}})(\vZ^k_{\star{2}})^T+\vZ^k(\vZ^k-\vZ^k_{\star{2}})^T)\|_F\\
        &\leq(\|\vZ^k_{\star{2}}\|+\|\vZ^k\|)\distP{\vZ^k}{\vZ_\star}
        \\
        &\lesssim\sigma_1(\vM_\star)^\frac{1}{2}\distP{\vZ^k}{\vZ_\star}.
    \end{align*}
The last inequality invokes results from the induction hypotheses \eqref{eq:distancebd} and \eqref{eq:nmcond2}, which are $ \|\vZ^k_{\star{2}}\|\lesssim\sigma_1\(\vM_\star\)^{\frac{1}{2}}$ and thus $\|\vZ^k\|\leq \|\vZ^k-\vZ^k_{\star{2}}\|+\|\vZ^k_{\star{2}}\|\lesssim\sigma_1\(\vM_\star\)^{\frac{1}{2}}$. Therefore, $\|\vx^k-\vx\|_2\leq O(\varepsilon)$ when $\distP{\vZ^k}{\vZ_\star}\leq \varepsilon$.

\subsection{Proof of Lemma~\ref{lem:relations-dist}} \label{apd:pf-lemma1-relatdist}
{
 From the definitions of $\mbox{dist}^2_Q(\vZ,\vZ_\star)$ and $\mbox{dist}^2_P(\vZ,\vZ_\star)$ 
\begin{align*}
  2\mbox{dist}^2_Q(\vZ,\vZ_\star)&\overset{(a)}{\geq} \inf_ {\substack{\vQ\in\C^{r\times r},\\\text{invertible}}}\ln\vZ-\vZ_\star\vQ\rn_F^2+\ln\vZ-\vZ_\star\vQ^{-T}\rn_F^2
  \\&=\mbox{dist}^2_P(\vZ,\vZ_\star),
\end{align*}
where $(a)$ results from that the domain becomes larger and the first part is proven.  

Then we prove the second part.  For  any invertible $\vP$ such that  
 $\sqrt{\ln\vZ-\vZ_\star\vP\rn_F^2+\ln\vZ-\vZ_\star\vP^{-T}\rn_F^2}\leq\varepsilon\sigma_r(\vM_\star)^{\frac{1}{2}}$, one has 
     \begin{align*}
         \|\vZ_\star(\vP-\vP^{-T})\|_F&\leq\|\vZ-\vZ_\star\vP\|_F+\|\vZ-\vZ_\star\vP^{-T}\|_F
         \\&\leq2\varepsilon\sigma_r(\vM_\star)^{\frac{1}{2}}, 
     \end{align*}
     where we invoke the triangle inequality. From the relation that $\|\vA\vB\|_F\geq\sigma_r(\vB)\|\vA\|_F$, we have
     \begin{align*}
         \sigma_r(\vM_\star)^{\frac{1}{2}}\|\vP-\vP^{-T}\|_F\leq \|\vZ_\star(\vP-\vP^{-T})\|_F,
     \end{align*}
     therefore $\|\vP-\vP^{-T}\|_F\leq2\varepsilon$, and  under this condition, $\mbox{dist}^2_P(\vZ,\vZ_\star)$ can be reformulated as
     \begin{align*}
    \mbox{dist}^2_P(\vZ,\vZ_\star){=}\inf_ {\substack{\|\vP-\vP^{-T}\|_F\leq2\varepsilon\\\vP\in\C^{r\times r}, \text{invertible}}}\ln\vZ{-}\vZ_\star\vP\rn_F^2{+}\ln\vZ{-}\vZ_\star\vP^{-T}\rn_F^2. 
     \end{align*}
     When $\varepsilon=0$, one has $\vP=\vP^{-T}$ and this means $\vP\vP^T=\vI$ where $\vP$ is a complex orthogonal matrix, then we obtain  
      \begin{align*}
    \mbox{dist}^2_P(\vZ,\vZ_\star)&{=}\inf_ {\substack{\|\vP-\vP^{-T}\|_F=0\\\vP\in\C^{r\times r}, \text{invertible}}}\ln\vZ{-}\vZ_\star\vP\rn_F^2{+}\ln\vZ{-}\vZ_\star\vP^{-T}\rn_F^2
    \\&=2\mbox{dist}^2_Q(\vZ,\vZ_\star),
     \end{align*}
     that is, the distance metric $\mbox{dist}_P^2(\vZ,\vZ_\star)$ is equivalent to the distance metric $\mbox{dist}_Q^2(\vZ,\vZ_\star)$. According to the linear convergence of $\mbox{dist}_P(\vZ,\vZ_\star)$ in Theorem \ref{thm:recovery_guarantee}, we have the linear convergence of $\varepsilon$ to zero. As a result, $\mbox{dist}_P^2(\vZ,\vZ_\star)$ reduces to $\mbox{dist}_Q^2(\vZ,\vZ_\star)$ asymptotically as $\varepsilon \to 0$.
     }
\section{}  \label{apd:suppl}

\subsection{ Proof of Lemma~\ref{lem:existence}}\label{apd:subsec_pf_existence}

Obviously, one can establish the equivalence between the following two minimization problems
\begin{align*}
    &\inf_{\substack{\vP \in \C^{r\times r} \\ \text { invertible }}} \sqrt{ \|\vZ-\vZ_{\star}\vP\|_F^2+\|\vZ-\vZ_{\star}\vP^{-T}\|_F^2}\numberthis\label{eq:dist_inf_1}\\
    &=\inf_{\substack{\vR \in \C^{r\times r} \\ \text { invertible }}} \sqrt{\|\vZ-\vZ_{\star}\hat{\vP}\vR\|_F^2+\|\vZ-\vZ_{\star}\hat{\vP}^{-T}\vR^{-T}\|_F^2}. \numberthis\label{eq:dist_inf_2}
\end{align*}
If the minimizer of the second optimization problem \eqref{eq:dist_inf_2} is attained at some $\vR$, then $\hat{\vP}\vR$ must be the minimizer of the first problem \eqref{eq:dist_inf_1}.  

As  $\hat{\vZ}_{\star{1}}=\vZ_{\star}\hat{\vP}$ and $\hat{\vZ}_{\star{2}}=\vZ_{\star}\hat{\vP}^{-T}$, we have
\begin{align*}
 &\inf_{\substack{\vR \in \C^{r\times r} \\ \text { invertible }}} \sqrt{\|\vZ-\hat{\vZ}_{\star{1}}\vR\|_F^2+\|\vZ-\hat{\vZ}_{\star{2}}\vR^{-T}\|_F^2}\\
 &\leq \sqrt{\|\vZ-\hat{\vZ}_{\star{1}}\|_F^2+\|\vZ-\hat{\vZ}_{\star{2}}\|_F^2}.
\end{align*}

 For any $\vR$ such that $\sqrt{\|\vZ-\hat{\vZ}_{\star{1}}\vR\|_F^2+\|\vZ-\hat{\vZ}_{\star{2}}\vR^{-T}\|_F^2}\leq \sqrt{\|\vZ-\hat{\vZ}_{\star{1}}\|_F^2+\|\vZ-\hat{\vZ}_{\star{2}}\|_F^2}$, we have  
\begin{align*}
   \|\hat{\vZ}_{\star{1}}(\vI-\vR)\|_F&\leq \|\vZ-\hat{\vZ}_{\star{1}}\|_F+\|\vZ-\hat{\vZ}_{\star{1}}\vR\|_F\\
    &\leq 2\varepsilon\sigma_r(\vM_\star)^\frac{1}{2},
\end{align*}
where we use the fact $\sqrt{\|\vZ-\hat{\vZ}_{\star{1}}\|_F^2+\|\vZ-\hat{\vZ}_{\star{2}}\|_F^2}\leq\varepsilon\sigma_r(\vM_\star)^\frac{1}{2}$ in the last inequality. When $
\sigma_r(\hat{\vZ}_{\star{1}}) \geq \frac{\sqrt{2}}{2}\sigma_r(\vM_\star)^{\frac{1}{2}}
  $, we further obtain that 
  \begin{align*}
    \|\hat{\vZ}_{\star{1}}(\vI-\vR)\|_F\geq\frac{\sqrt{2}}{2}\sigma_r(\vM_\star)^{\frac{1}{2}}\|\vI-\vR\|_F.
\end{align*}

Thus we obtain that $\|\vI-\vR\|_F\leq 2\sqrt{2}\varepsilon$. From Weyl's inequality we see that $\sigma_r(\vR)\geq 1-2\sqrt{2}\varepsilon\geq\frac{1}{2}$ as long as $\varepsilon\leq\frac{\sqrt{2}}{8}$, which denotes $\vR$ is invertible. 

Therefore  the minimization problem \eqref{eq:dist_inf_2} is equivalent to 
\begin{align*}
    \min_{\substack{\vR \in \C^{r\times r} \\ \text { invertible }}}&\sqrt{\|\vZ-\vZ_{\star}\hat{\vP}\vR\|_F^2+\|\vZ-\vZ_{\star}\hat{\vP}^{-T}\vR^{-T}\|_F^2}\\
    \mathrm{s.t.} ~~~~& \|\vI-\vR\|_F\leq 2\sqrt{2}\varepsilon,
\end{align*}
which is a continuous optimization problem over a compact set of invertible matrices. The minimizer exists from Weierstrass extreme value theorem. Similarly, we can prove the existence of minimizer if we have $\sigma_r(\hat{\vZ}_{\star{2}})\geq \frac{\sqrt{2}}{2}\sigma_r(\vM_\star)^{\frac{1}{2}}
$.
 
\subsection{Proof of Lemma~\ref{lem:symdist_inq}}\label{apd:proof-lem4-syminq}
 
We introduce the following special form of orthogonal Procrustes problem firstly 
\begin{align*}
\min_{\vQ \in \mathcal{O}}\ln\vZ-\vZ_\star\vQ\rn_F. 
\end{align*}
Let the SVD of $\Real(\vZ_\star^H\vZ)$ being $\Real(\vZ_\star^H\vZ)=\vQ_1\bm{\Lambda}\vQ_2^T$, this problem has a closed-form solution as $\vQ_{\vz}=\vQ_1\vQ_2^T.$

Let $\vF=\begin{bmatrix}\vZ\\\bar{\vZ}\end{bmatrix}$, $\vF_\star=\begin{bmatrix}\vZ_\star\\\bar{\vZ}_\star\end{bmatrix}$, $\vDelta=\vF\vQ_{\vz}^T-\vF_{\star}$,  and we know $\Real(\vZ_\star^H\vZ)=\frac{1}{2}\vF_\star^H\vF=\vQ_1\bm{\Lambda}\vQ_2^T$. As  $\vQ_{\vz}{=\vQ_1\vQ_2^T}$, $\vF_\star^H\vF\vQ_{\vz}^T$ is a positive semi-definite real matrix and this means $\vF_\star^H\vDelta=\vF_\star^H\vF\vQ_{\vz}^T-\vF_\star^H\vF_\star=\vQ_{\vz}\vF^H\vF_\star-\vF_\star^H\vF_\star=\vDelta^H\vF_\star$. One can establish
\begin{align*}
&\ln\vF\vF^H-\vF_\star\vF_\star^H\rn_F^2\\
&=\ln\vF_\star\vDelta^H+\vDelta\vF_\star^H+\vDelta\vDelta^H\rn_F^2\\
&=\operatorname{tr}(2\vF_\star^H\vF_\star\vDelta^H\vDelta+(\vDelta^H\vDelta)^2+(\vF_\star^H\vDelta)^2+(\vDelta^H\vF_\star)^2\\&\quad+2\vF_\star^H\vDelta\vDelta^H\vDelta+2\vDelta^H\vF_\star\vDelta^H\vDelta)\\
&\overset{(a)}{=}\operatorname{tr}(2\vF_\star^H\vF_\star\vDelta^H\vDelta+(\vDelta^H\vDelta+\sqrt{2}\vF_\star^H\vDelta)^2\\
&\quad+(4-2\sqrt{2})\vF_\star^H\vDelta\vDelta^H\vDelta)\\
&=\operatorname{tr}(2(\sqrt{2}-1)\vF_\star^H\vF_\star\vDelta^H\vDelta+(\vDelta^H\vDelta+\sqrt{2}\vF_\star^H\vDelta)^2\\
&\quad+(4-2\sqrt{2})\vF_\star^H\vF\vQ_{\vz}^T\vDelta^H\vDelta)\\
&\overset{(b)}{\geq}4(\sqrt{2}-1)\|(\vF\vQ_{\vz}^T-\vF_\star)\vSS_\star^{\frac{1}{2}}\|_F^2\\
&\geq8(\sqrt{2}-1)\sigma_r(\vM_\star)\underset{\vQ \in \mathcal{O}}{\operatorname{min}}\ln\vZ-\vZ_\star\vQ\rn_F^2  \\&\overset{(c)}{\geq}
4(\sqrt{2}-1)\sigma_r(\vM_\star)\distPsq{\vZ}{\vZ_\star}.\numberthis\label{eq:univ_alignment_lowerbd}
\end{align*}
 Step $(a)$ results from $\vF_\star^H\vDelta=\vDelta^H\vF_\star$. Step $(b)$ follows from the fact $\operatorname{tr}(\vF_\star^H\vF\vQ_{\vz}^T\vDelta^H\vDelta)\geq0$ ($\vDelta^H\vDelta$ and $\vF_\star^H\vF\vQ_{\vz}^T$ is positive semi-definite) and the fact $(\vDelta^H\vDelta+\sqrt{2}\vF_\star^H\vDelta)^H=\vDelta^H\vDelta+\sqrt{2}\vF_\star^H\vDelta$.  
Step $(c)$ invokes the fact that 
\begin{align*}
\distPsq{\vZ}{\vZ_\star}&=\inf_{\substack{\vP \in \C^{r\times r} \\ \text { invertible }}}\ln\vZ-\vZ_\star\vP\rn_F^2+\ln\vZ-\vZ_\star\vP^{-T}\rn_F^2\\&\leq\ln\vZ-\vZ_\star\vQ_{\vz}\rn_F^2+\ln\vZ-\vZ_\star\vQ_{\vz}^{-T}\rn_F^2
\\&=2\underset{\vQ \in \mathcal{O}}{\operatorname{min}}\ln\vZ-\vZ_\star\vQ\rn_F^2. 
\end{align*}
We establish the following inequality from some technical derivations
\begin{align*}
&\ln\vZ\vZ^H-\vZ_\star\vZ_\star^H\rn_F^2+\ln\bar{\vZ}\bar{\vZ}^H-\bar{\vZ}_\star\bar{\vZ}_\star^H\rn_F^2\\
&\quad-2\ln\vZ\vZ^T-\vZ_\star\vZ_\star^T\rn_F^2\\
&=\(\ln\vZ\vZ^H\rn_F^2+\ln\bar{\vZ}\bar{\vZ}^H\rn_F^2-2\ln\vZ\vZ^T\rn_F^2\)\\
&\quad+\(\ln\vZ_\star\vZ_\star^H\rn_F^2+\ln\bar{\vZ}_\star\bar{\vZ}_\star^H\rn_F^2-2\ln\vZ_\star\vZ_\star^T\rn_F^2\)\\&\quad-2\(\ln\vZ^H\vZ_\star\rn_F^2+\ln\bar{\vZ}^H\bar{\vZ}_\star\rn_F^2-2\Real\la\vZ\vZ^T,\vZ_\star\vZ_\star^T\ra\)\\&=-2\ln\vZ^H\vZ_\star-\bar{\vZ}^H\bar{\vZ}_\star\rn_F^2\leq0,
\end{align*}
where the second equality is from $\ln\vZ\vZ^T\rn_F^2=\ln\vZ\vZ^H\rn_F^2=\ln\bar{\vZ}\bar{\vZ}^H\rn_F^2=\ln\vSS\rn_F^2$ and $\ln\vZ_\star\vZ_\star^T\rn_F^2=\ln\vZ_\star\vZ_\star^H\rn_F^2=\ln\bar{\vZ}_\star\bar{\vZ}_\star^H\rn_F^2=\ln\vSS_\star\rn_F^2$. The previous inequality helps establish the upper-bound of $\ln\vF\vF^H-\vF_\star\vF_\star^H\rn_F^2$, which is
\begin{align*}
&\ln\vF\vF^H-\vF_\star\vF_\star^H\rn_F^2\\
&=2\ln\vZ\vZ^T-\vZ_\star\vZ_\star^T\rn_F^2+\ln\vZ\vZ^H-\vZ_\star\vZ_\star^H\rn_F^2\\
&\quad+\ln\bar{\vZ}\bar{\vZ}^H-\bar{\vZ}_\star\bar{\vZ}_\star^H\rn_F^2\\
&\leq4\ln\vZ\vZ^T-\vZ_\star\vZ_\star^T\rn_F^2=4\ln\vM-\vM_\star\rn_F^2. \numberthis\label{eq:univ_Csym_upbd}
\end{align*}
combining \eqref{eq:univ_alignment_lowerbd} and \eqref{eq:univ_Csym_upbd}, we complete the proof of Lemma~\ref{lem:symdist_inq}.

\subsection{Proof of Lemma~\ref{lem:Hankcom_tr_inq}}\label{apd:subsec_hankelcomp_tr_inq}

   Let $\vH_{a}$ denotes the $n_s\times n_s$ Hankel basis matrix with entries in the $a$-th skew diagonal as $\frac{1}{\sqrt{w_a}}$. Since $p=\frac{m}{n}$, we have
   \begin{align*}
       &\la\G(p^{-1}\P_{{\hat{\Omega}}}-\I)\G^*(\vA\vB^T),\vC\vD^T\ra\\
       &=\sum_{k=1}^{m}\Big(\frac{n}{m}\langle\vA\vB^T,\vH_{a_k}\rangle\langle\vH_{a_k},\vC\vD^T\rangle\\&\quad-\frac{1}{m}\sum_{a=0}^{n-1}\langle\vA\vB^T,\vH_a\rangle\langle\vH_a,\vC\vD^T\rangle\Big)=\sum_{k=1}^m z_{a_k},
   \end{align*}

 It is obviously that $z_{a_k}$ are i.i.d random variables and $\E[z_{a_k}]=0$. 
 We first point out some results   \begin{align*}  |\langle\vA\vB^T,\vH_a\rangle|&=|\langle\vA,\vH_a\bar{\vB}\rangle|\leq\|\vA\|_F\|\vB\|_{2,\infty},\\
|\langle\vA\vB^T,\vH_a\rangle|&=|\langle\vB^T,\vA^H\vH_a\rangle|\leq\|\vA\|_{2,\infty}\|\vB\|_F.
   \end{align*}  
where we use the fact that $\|\vH_a\bar{\vB}\|_F\leq\|\vB\|_{2,\infty}$ and $\|\vA^H\vH_a\|_F\leq\|\vA\|_{2,\infty}$. Similar results hold for $|\langle\vH_a,\vC\vD^T\rangle|$. Thus we have
   \begin{align*}
       &|z_{a_k}|\leq \frac{2n}{m}\max_{a}|\langle\vA\vB^T,\vH_a\rangle\langle\vH_a,\vC\vD^T\rangle|\leq \frac{2n}{m}L,
   \end{align*}
   where
   \begin{align*}
     L=&\min\{\ln\vA\rn_F\ln\vB\rn_{2,\infty},\ln\vA\rn_{2,\infty}\ln\vB\rn_F\}\\&\cdot\min\{\ln\vC\rn_F\ln\vD\rn_{2,\infty},\ln\vC\rn_{2,\infty}\ln\vD\rn_F\}. 
   \end{align*}
As $z_{a_k}$ is a complex scalar, we have 
\begin{small}
   \begin{align*}
       \E\({z_{a_k}^H}z_{a_k}\)=\E(z_{a_k}z_{a_k}^H)&=\E(\frac{n^2}{m^2}|\langle\vA\vB^T,\vH_{a_k}\rangle|^2|\langle\vH_{a_k},\vC\vD^T\rangle|^2)\\
       &\quad-\frac{1}{m^2}|\sum_{a=0}^{n-1}\langle\vA\vB^T,\vH_a\rangle\langle\vH_a,\vC\vD^T\rangle|^2\\
       &\leq \E(\frac{n^2}{m^2}|\langle\vA\vB^T,\vH_{a_k}\rangle|^2|\langle\vH_{a_k},\vC\vD^T\rangle|^2).
   \end{align*}
   \end{small}Therefore
\begin{align*}
\sum_{k=1}^{m}\E\(z_{a_k}z_{a_k}^H\)=\sum_{k=1}^{m}\E\(z_{a_k}^Hz_{a_k}\)\leq\frac{n^2}{m}L^2.
\end{align*}
Invoking Bernstein's inequality \cite[Theorem~1.6]{Tropp2012}, we have
\begin{align*}
\mathbb{P}\(\left|\sum_{k=1}^m z_{a_k}\right|\geq t\)\leq 2\exp\(\frac{-mt^2/2}{n^2L^2+2nLt/3}\).
\end{align*}
Consequently when $m\gtrsim O(\log(n))$, $\left|\sum_{k=1}^m z_{a_k}\right|\leq c_4L\sqrt{\frac{n^2\log(n)}{m}}$ holds with probability at least $1-n^{-2}$.


\begin{thebibliography}{10}
	
	\bibitem{Potter2010}
	L.~Potter, E.~Ertin, J.~Parker, and M.~Cetin, ``Sparsity and compressed sensing
	in radar imaging,'' {\em Proc. IEEE}, vol.~98, no.~6, pp.~1006--1020, 2010.
	
	\bibitem{Tropp2010}
	J.~A. Tropp, J.~N. Laska, M.~F. Duarte, J.~K. Romberg, and R.~G. Baraniuk,
	``Beyond nyquist: Efficient sampling of sparse bandlimited signals,'' {\em
		IEEE Trans. Inf. Theory}, vol.~56, no.~1, pp.~520--544, 2010.
	
	\bibitem{Lustig2007}
	M.~Lustig, D.~Donoho, and J.~M. Pauly, ``Sparse {MRI}: The application of
	compressed sensing for rapid {MR} imaging,'' {\em Magn. Reson. Med.},
	vol.~58, no.~6, pp.~1182--1195, 2007.
	
	\bibitem{Paredes2007}
	J.~L. Paredes, G.~R. Arce, and Z.~Wang, ``Ultra-wideband compressed sensing:
	Channel estimation,'' {\em IEEE J. Sel. Top. Signal Process.}, vol.~1, no.~3,
	pp.~383--395, 2007.
	
	\bibitem{Chen2014}
	Y.~Chen and Y.~Chi, ``Robust spectral compressed sensing via structured matrix
	completion,'' {\em IEEE Trans. Inf. Theory}, vol.~60, no.~10, pp.~6576--6601,
	2014.
	
	\bibitem{Cai2018}
	J.-F. Cai, T.~Wang, and K.~Wei, ``Spectral compressed sensing via projected
	gradient descent,'' {\em SIAM J. Optim.}, vol.~28, no.~3, pp.~2625--2653,
	2018.
	
	\bibitem{Cai2019}
	J.-F. Cai, T.~Wang, and K.~Wei, ``Fast and provable algorithms for spectrally
	sparse signal reconstruction via low-rank hankel matrix completion,'' {\em
		Appl. Comput. Harmon. Anal.}, vol.~46, no.~1, pp.~94--121, 2019.
	
	\bibitem{Zhang2021}
	X.~Zhang, Y.~Liu, and W.~Cui, ``Spectrally sparse signal recovery via hankel
	matrix completion with prior information,'' {\em IEEE Trans. Signal
		Process.}, vol.~69, pp.~2174--2187, 2021.
	
	\bibitem{Tang2013}
	G.~Tang, B.~N. Bhaskar, P.~Shah, and B.~Recht, ``Compressed sensing off the
	grid,'' {\em IEEE Trans. Inf. Theory}, vol.~59, no.~11, pp.~7465--7490, 2013.
	
	\bibitem{Horn2012}
	R.~A. Horn and C.~R. Johnson, {\em Matrix Analysis}.
	\newblock USA: Cambridge University Press, 2nd~ed., 2012.
	
	\bibitem{Horn1999}
	R.~A. Horn and D.~I. Merino, ``The jordan canonical forms of complex orthogonal
	and skew-symmetric matrices,'' {\em Linear Algebra Appl.}, vol.~302-303,
	pp.~411--421, 1999.
	
	\bibitem{Luk1997}
	F.~T. Luk and S.~Qiao, ``Using complex-orthogonal transformations to
	diagonalize a complex symmetric matrix,'' in {\em Opt. Photon.}, 1997.
	
	\bibitem{Didenko1978}
	V.~D. Didenko and V.~A. Chernetskii, ``The riemann boundary problem with a
	complex orthogonal matrix,'' {\em Math. Notes Acad. Sci. USSR}, vol.~23,
	pp.~220--227, Mar. 1978.
	
	\bibitem{Candes2006}
	E.~Cand\`es, J.~Romberg, and T.~Tao, ``Robust uncertainty principles: exact
	signal reconstruction from highly incomplete frequency information,'' {\em
		IEEE Trans. Inf. Theory}, vol.~52, no.~2, pp.~489 -- 509, 2006.
	
	\bibitem{Donoho2006}
	D.~Donoho, ``Compressed sensing,'' {\em IEEE Trans. Inf. Theory}, vol.~52,
	no.~4, pp.~1289 --1306, 2006.
	
	\bibitem{Shen2022}
	J.~Shen, J.-S. Chen, H.-D. Qi, and N.~Xiu, ``A penalized method of alternating
	projections for weighted low-rank hankel matrix optimization,'' {\em Math.
		Program. Comput.}, vol.~14, no.~3, pp.~417--450, 2022.
	
	\bibitem{Zheng2016}
	Q.~Zheng and J.~Lafferty, ``Convergence analysis for rectangular matrix
	completion using {B}urer-{M}onteiro factorization and gradient descent,''
	{\em arXiv:1605.07051}, 2016.
	
	\bibitem{Xu2008}
	W.~Xu and S.~Qiao, ``A fast symmetric {SVD} algorithm for square hankel
	matrices,'' {\em Linear Algebra Appl.}, vol.~428, no.~2-3, pp.~550--563,
	2008.
	
	\bibitem{Andersson2011}
	F.~Andersson, M.~Carlsson, and P.-A. Ivert, ``A fast alternating projection
	method for complex frequency estimation,'' {\em arXiv:1107.2028}, 2011.
	
	\bibitem{Wei2016}
	K.~Wei, J.-F. Cai, T.~F. Chan, and S.~Leung, ``Guarantees of riemannian
	optimization for low rank matrix recovery,'' {\em SIAM J. Matrix Anal.
		Appl.}, vol.~37, no.~3, pp.~1198--1222, 2016.
	
	\bibitem{Wei2020}
	K.~Wei, J.-F. Cai, T.~F. Chan, and S.~Leung, ``Guarantees of riemannian
	optimization for low rank matrix completion,'' {\em Inverse Probl. Imag.},
	vol.~14, no.~2, pp.~233--265, 2020.
	
	\bibitem{Wu2023}
	X.~Wu, Z.~Yang, J.-F. Cai, and Z.~Xu, ``Spectral super-resolution on the unit
	circle via gradient descent,'' in {\em {IEEE} Int. Conf. Acoust., Speech,
		Signal Process.}, 2023.
	
	\bibitem{Mao2022}
	S.~Mao and J.~Chen, ``Blind super-resolution of point sources via projected
	gradient descent,'' {\em IEEE Trans. Signal Process.}, vol.~70,
	pp.~4649--4664, 2022.
	
	\bibitem{Yin2023}
	J.~Yin, S.~Peng, Z.~Yang, B.~Chen, and Z.~Lin, ``Hypergraph based
	semi-supervised symmetric nonnegative matrix factorization for image
	clustering,'' {\em Pattern Recogn.}, vol.~137, p.~109274, May 2023.
	
	\bibitem{He2011}
	Z.~He, S.~Xie, R.~Zdunek, G.~Zhou, and A.~Cichocki, ``Symmetric nonnegative
	matrix factorization: Algorithms and applications to probabilistic
	clustering,'' {\em IEEE Trans. Neural Netw.}, vol.~22, no.~12,
	pp.~2117--2131, 2011.
	
	\bibitem{Chen2015}
	Y.~Chen and M.~J. Wainwright, ``Fast low-rank estimation by projected gradient
	descent: General statistical and algorithmic guarantees,'' {\em
		arxiv:1509.03025}, 2015.
	
	\bibitem{Comon2008}
	P.~Comon, G.~Golub, L.-H. Lim, and B.~Mourrain, ``Symmetric tensors and
	symmetric tensor rank,'' {\em SIAM J. Matrix Anal. Appl}, vol.~30,
	pp.~1254--1279, Jan. 2008.
	
	\bibitem{Brachat2010}
	J.~Brachat, P.~Comon, B.~Mourrain, and E.~Tsigaridas, ``Symmetric tensor
	decomposition,'' {\em Linear Algebra Appl.}, vol.~433, pp.~1851--1872, Dec.
	2010.
	
	\bibitem{Cai2015a}
	J.~Cai, M.~Baskaran, B.~Meister, and R.~Lethin, ``Optimization of symmetric
	tensor computations,'' in {\em IEEE High Perform. Extreme Comput. Conf.},
	pp.~1--7, 2015.
	
	\bibitem{Candes2009}
	E.~J. Cand{\`e}s and B.~Recht, ``Exact matrix completion via convex
	optimization,'' {\em Found. Comput. Math.}, vol.~9, no.~6, pp.~717--772,
	2009.
	
	\bibitem{Liao2016}
	W.~Liao and A.~Fannjiang, ``{MUSIC} for single-snapshot spectral estimation:
	Stability and super-resolution,'' {\em Appl. Comput. Harmon. Anal.}, vol.~40,
	no.~1, pp.~33--67, 2016.
	
	\bibitem{Cherapanamjeri2017}
	Y.~Cherapanamjeri, K.~Gupta, and P.~Jain, ``Nearly optimal robust matrix
	completion,'' in {\em Proc. Int. Conf. Mach. Learn.}, pp.~797--805, 2017.
	
	\bibitem{Jain2013}
	P.~Jain, P.~Netrapalli, and S.~Sanghavi, ``Low-rank matrix completion using
	alternating minimization,'' in {\em Proc. 45th Annu. ACM symp. Theory
		Comput.}, pp.~665--674, {ACM}, 2013.
	
	\bibitem{Zhang2018}
	S.~Zhang, Y.~Hao, M.~Wang, and J.~H. Chow, ``Multichannel hankel matrix
	completion through nonconvex optimization,'' {\em IEEE J. Sel. Top. Signal
		Process.}, vol.~12, no.~4, pp.~617--632, 2018.
	
	\bibitem{Zhang2019}
	S.~Zhang and M.~Wang, ``Correction of corrupted columns through fast robust
	hankel matrix completion,'' {\em IEEE Trans. Signal Process.}, vol.~67,
	no.~10, pp.~2580--2594, 2019.
	
	\bibitem{Chebotarev2014}
	A.~M. Chebotarev and A.~E. Teretenkov, ``Singular value decomposition for the
	takagi factorization of symmetric matrices,'' {\em Appl. Math. Comput.},
	vol.~234, pp.~380--384, 2014.
	
	\bibitem{Ikramov2012}
	K.~D. Ikramov, ``Takagi's decomposition of a symmetric unitary matrix as a
	finite algorithm,'' {\em Comp. Math. Math. Phys.}, vol.~52, no.~1, pp.~1--3,
	2012.
	
	\bibitem{Tu2016}
	S.~Tu, R.~Boczar, M.~Simchowitz, M.~Soltanolkotabi, and B.~Recht, ``Low-rank
	solutions of linear matrix equations via procrustes flow,'' in {\em Proc.
		Int. Conf. Mach. Learn.}, pp.~964--973, 2016.
	
	\bibitem{Ma2019}
	C.~Ma, K.~Wang, Y.~Chi, and Y.~Chen, ``Implicit regularization in nonconvex
	statistical estimation: Gradient descent converges linearly for phase
	retrieval, matrix completion, and blind deconvolution,'' {\em Found. Comput.
		Math.}, vol.~20, no.~3, pp.~451--632, 2019.
	
	\bibitem{Vershynin2018}
	R.~Vershynin, {\em High-Dimensional Probability}.
	\newblock Cambridge University Press, 2018.
	
	\bibitem{Sanson2020}
	J.~B. Sanson, P.~M. Tome, D.~Castanheira, A.~Gameiro, and P.~P. Monteiro,
	``High-resolution delay-doppler estimation using received communication
	signals for {OFDM} radar-communication system,'' {\em IEEE Trans. Veh.
		Technol.}, vol.~69, pp.~13112--13123, nov 2020.
	
	\bibitem{Nikookar}
	H.~Nikookar and R.~Prasad, ``Multicarrier transmission with nonuniform carriers
	in a multipath channel,'' in {\em Proc. Int. Conf. Univ. Pers. Commun.},
	{IEEE}, 1996.
	
	\bibitem{Yang2017}
	Q.~Yang, Z.~Wang, and Q.~Huang, ``An efficient non-uniform multi-tone system
	based on ramanujan sums,'' in {\em Proc. 2nd Joint Int. Inf. Tech. Mech.
		Electron. Eng. Conf.}, Atlantis Press, 2017.
	
	\bibitem{Berger2010}
	C.~R. Berger, B.~Demissie, J.~Heckenbach, P.~Willett, and S.~Zhou, ``Signal
	processing for passive radar using ofdm waveforms,'' {\em IEEE J. Sel. Top.
		Signal Process.}, vol.~4, no.~1, pp.~226--238, 2010.
	
	\bibitem{Ma2021}
	C.~Ma, Y.~Li, and Y.~Chi, ``Beyond procrustes: Balancing-free gradient descent
	for asymmetric low-rank matrix sensing,'' {\em IEEE Trans. Signal Process.},
	vol.~69, pp.~867--877, 2021.
	
	\bibitem{Recht2011}
	B.~Recht, ``A simpler approach to matrix completion,'' {\em J. Mach. Learn.
		Res.}, vol.~12, p.~3413–3430, 2011.
	
	\bibitem{Tropp2012}
	J.~A. Tropp, ``User-friendly tail bounds for sums of random matrices,'' {\em
		Found. Comput. Math.}, vol.~12, no.~4, pp.~389--434, 2012.
	
\end{thebibliography}
\end{document}